%% file: main-artcl.tex
\documentclass[a4paper]{article}

\usepackage{amsmath,amsthm}
\usepackage{url}

\input{packages-extra}

\input{title-artcl}

\input{macros}

\renewcommand{\subparagraph}[1]{\paragraph{#1}}

\newtheorem{theorem}{Theorem}
\newtheorem{example}[theorem]{Example}
\newtheorem{proposition}[theorem]{Proposition}
\newtheorem{lemma}[theorem]{Lemma}
\newtheorem{corollary}[theorem]{Corollary}
\newtheorem{remark}[theorem]{Remark}

\begin{document}

\maketitle

\input{sections}

\end{document}

%% file: packages-extra.tex
\usepackage{microtype}
\usepackage{utf8-math}
\usepackage{xspace}
\usepackage{thm-restate} 
\usepackage{pifont}
\usepackage{enumerate}
\usepackage[pdftex,
bookmarks,
             hidelinks
           ]{hyperref}
\usepackage{tikz}
\usetikzlibrary{matrix}

%% file: title-artcl.tex
\def\thetitle{On Freeze LTL with Ordered Attributes}

\title{\thetitle%
{\footnotetext{%
\scriptsize{This is a free and extended version of an article published in the proceedings of FoSSaCS 2016. 
The work was partially supported by EGIDE/DAAD-Procope (FREQS).}}}%
}

\author{Normann Decker and Daniel Thoma\\[1em]
Institute for Software Engineering and Programming Languages\\
           Universität zu Lübeck, Germany\\
  \url{{decker,thoma}@isp.uni-luebeck.de}}
\date{}

\hypersetup{
  pdftitle={\thetitle},
  pdfauthor={Normann Decker and Daniel Thoma}
}

%% file: macros.tex
\newcommand{\ignore}[1]{}

\newcommand{\LTL}{\textrm{LTL}\xspace}

\newcommand{\cl}{\mathsf{cl}}
\newcommand{\dom}{\mathsf{dom}}

\newcommand{\xmark}{\text{\ding{55}}}
\newcommand{\cmark}{\text{\ding{51}}}
\newcommand{\add}{\mathsf{add}}
\newcommand{\nxt}{\mathsf{next}}
\newcommand{\stp}{\mathsf{setup}}
\newcommand{\cpy}{\mathsf{copy}}

\newcommand{\st}{\mathsf{stor}}
\newcommand{\aux}{\mathsf{aux}}

\newcommand{\sub}{\mathsf{sub}}

\DeclareUnicodeCharacter{2193}{{\downarrow}} 
\DeclareUnicodeCharacter{2191}{{\uparrow}} 

\newcommand{\ol}{\overline}

\newcommand{\cp}{\texttt{cp}\xspace}
\renewcommand{\min}{\texttt{min}\xspace}

\newcommand{\U}{\ensuremath{\operatorname{U}}\xspace}

\newcommand{\R}{\ensuremath{\operatorname{R}}\xspace}
\newcommand{\X}{\ensuremath{\operatorname{X}}\xspace}
\newcommand{\WX}{\ensuremath{\operatorname{\overline{X}}}\xspace}

\newcommand{\G}{\ensuremath{\operatorname{G}}\xspace}
\newcommand{\F}{\ensuremath{\operatorname{F}}\xspace}

\newcommand{\true}{\mathrm{true}}
\newcommand{\false}{\mathrm{false}}

%% file: sections.tex
\input{abstract}

\input{sec_introduction}

\input{sec_freeze}

\input{sec_freeze_undecidability}

\input{sec_ncs}
\input{sec_ncs_lowerbound}

\input{sec_freeze_decidability}

\input{sec_conclusion}

\bibliography{bibliography-auto,bibliography-extra}

\appendix

\include{apx_freeze_undecidability}

\include{apx_ncs_upperbound}

\include{apx_ncs_lowerbound}

\include{apx_freeze_linearisation}

\include{apx_freeze2ncs}

\include{apx_ncs2freeze}

%% file: abstract.tex
\begin{abstract}
\noindent 
This paper is concerned with Freeze LTL, a temporal logic on data words with registers.
In a (multi-attributed) data word each position carries a letter from a finite alphabet and 
assigns a data value to a fixed, finite set of attributes.
The satisfiability problem of Freeze LTL is undecidable if more than one register is available or tuples of data values can be stored and compared arbitrarily.
Starting from the decidable one-register fragment we propose an extension that allows for specifying a dependency relation on attributes.
This restricts in a flexible way how collections of attribute values can be stored and compared.
This conceptual dimension %
is orthogonal %
to the number of registers or the available temporal operators.
The extension is strict.
Admitting arbitrary dependency relations, satisfiability becomes undecidable.
Tree-like relations, however, induce a family of decidable fragments
escalating the ordinal-indexed hierarchy of fast-growing complexity classes, a recently introduced framework for non-primitive recursive complexities. %
This results in completeness for the class $𝐅_{ε_0}$.
We employ nested counter systems and show that they relate to the hierarchy in terms of the nesting depth.
\end{abstract}

%% file: sec_introduction.tex
\section{Introduction}
\label{sec:intro}

A central aspect in modern programming languages and software architectures is dynamic and unbounded creation of entities.
In particular object oriented designs rely on instantiation of objects on demand and flexible multi-threaded execution.
Finite abstractions can hardly reflect these dynamics and therefore infinite models are very valuable for specification and analysis.
This motivates us to study the theoretical framework of words over infinite alphabets.
It allows for abstracting, e.g., the internal structure and state of particular objects or processes while still being able to capture the architectural design in terms of interaction and relations between dynamically instantiated program parts.

These \emph{data words}, as we consider them here, are finite, non-empty sequences $w=(a_1,𝐝_1)(a_2,𝐝_2)…(a_n,𝐝_n)$ where the $i$-th position carries a letter $a_i$ from a finite alphabet $Σ$.
Additionally, for a fixed, finite set of \emph{attributes} $A$ a \emph{data valuation} $𝐝_i: A → Δ$ assigns to each attribute a \emph{data value} from an infinite domain $Δ$ with equality.

\subparagraph{Freeze LTL.}
In formal verification, temporal logics are widely used for formulating behavioural specifications and, regarding data, the concept of storing values in registers for comparison at different points in time is very natural.
This paper is therefore concerned with the logic \emph{Freeze LTL} \cite{DBLP:conf/time/DemriLN05} that extends classical Linear-time Temporal Logic (LTL) by registers and was extensively studied during the past decade.
Since the satisfiability problem of Freeze LTL is undecidable in general, we specifically consider the decidable fragment LTL$^↓_1$ \cite{DBLP:journals/tocl/DemriL09} that is restricted to a single register and future-time modalities.
More precisely, we propose a generalisation of this fragment and study the consequences in terms of decidability and complexity.

Considering specification and modelling, the ability of comparing \emph{tuples} of data values arbitrarily is a valuable feature.
  Unfortunately, this generically renders logics on data words undecidable (cf.\ related work below).
 We therefore extend Freeze LTL by a mechanism for carefully restricting the collections of values that can be compared in terms of a \emph{dependency relation} on attributes.
 In general, this does not suffice to regain decidability of the satisfiability problem.
 Imposing, however, a \emph{hierarchical} dependency structure such that comparison of attribute values is carried out in an ordered fashion, we obtain a strict hierarchy of decidable fragments parameterised by the maximal depth of the attribute hierarchy.

Before we exemplify this concept, let us introduce basic notation.
Let $Σ$ be a finite alphabet and $(A,⊑)$ a finite set of attributes together with a reflexive and transitive relation ${⊑}⊆A×A$, i.e., a \emph{quasi-ordering}, simply denoted $A$ if $⊑$ is understood.
We call our logic LTL$^↓_{qo}$ and define its \emph{syntax} according to the grammar
\[
  φ ::= a｜¬φ｜φ∧φ｜\Xφ｜φ\Uφ｜↓^xφ｜\ ↑^x
\]
for letters $a∈Σ$ and attributes $x∈A$.
We further include common abbreviations such as disjunction, implication or the temporal operators release ($φ\Rψ:= ¬(¬φ\U¬ψ)$), weak next ($\WXφ:=¬\X¬φ$) and globally ($\Gφ := \false\Rφ$).
The restriction of LTL$^↓_{qo}$ to a particular, fixed set of attributes $(A,⊑)$ is denoted LTL$^↓_{(A,⊑)}$ (or simply LTL$^↓_A$).

In the following, we explain the idea of our extension by means of an example. 
The formal semantics is defined in Section~\ref{sec:freeze}.

\begin{example} \label{ex:lock-res} Consider a system with arbitrarily many processes that can lock, unlock and use an arbitrary number of resources.
A data word over the alphabet $Σ=｛\mathsf{lock}, \mathsf{unlock}, \mathsf{use}, \mathsf{halt}｝$ 
can model its behaviour in terms of an interleaving of individual actions and global signals.
The corresponding data valuation can provide specific properties of an action, such as a unique identifier for the involved process and the resource.
Let us use attributes $A=｛\mathsf{pid}, \mathsf{res}｝$ and interpret data values from $Δ$ as IDs.
Notice that this way, we do not assume a bound on the number of involved entities.

Consider now the property that locked resources must not be used by foreign processes and all locks must be released on system halt.
To express this, we need to store both the process and resource ID for every \emph{lock} action and verify that a \emph{use} involving the same resource also involves the same process.
As mentioned earlier, employing a too liberal mechanism to store multiple data values at once breaks the possibility of automatic analysis.
In our case, however, we do not need to refer to processes independently.
It suffices to consider only resources individually and formulate that the particular process that locks a resource is the only one using it before unlocking.
This \emph{one-to-many} correspondence between processes and resources allows us to declare the attribute $\mathsf{pid}$ to be dependent on the attribute $\mathsf{res}$ and formulate the property by the formula 
\[
    \G(\mathsf{lock} → ↓^{\mathsf{pid}} ((\mathsf{use}\ ∧↑^{\mathsf{res}} →\ ↑^{\mathsf{pid}})∧ ¬\mathsf{halt} )\U (\mathsf{unlock}\ ∧ ↑^{\mathsf{pid}})).
\]
The \emph{freeze quantifier} $↓^\mathsf{pid}$ stores the current value assigned to $\mathsf{pid}$ and also implicitly that of all its dependencies, $\mathsf{res}$ in this case.
The \emph{check operator} $↑^x$, for an attribute $x∈A$, then verifies at some position that the current values of $x$ and its dependencies coincide with the information that was stored earlier.
Also, properties independent of the data can be verified within the same context, e.g., $¬\mathsf{halt}$ for preventing a shut down as long as any resource is still locked.
See Figure~\ref{fig:lock-res-words} for example words.

\begin{figure}[t]
\begin{center}
\begin{tikzpicture}
  \matrix[matrix of math nodes] (left) {
        & \mathsf{lock} & \mathsf{lock} & \mathsf{use} & \mathsf{unlock} & \mathsf{unlock} &[2em]
\mathsf{lock} & \mathsf{lock} & \mathsf{use} & \mathsf{unlock} & \mathsf{halt} & \mathsf{unlock} \\
    (\mathsf{res}) & 1 &2& 2& 2& 1 &[2em] 1 & 3& 3& 3& 9 &1\\
    (\mathsf{pid}) & 1 & 1& 1& 1& 1 &     1 & 2& 1& 1& 8 &1\\
  };
\end{tikzpicture}
\end{center}
\caption{The left word satisfies the formula from Example~\ref{ex:lock-res} whereas every strict prefix does not. The right word violates the property because at position 3 $\mathsf{use}$ holds and the value of $\mathsf{res}$ matches the one stored at position 2 but the whole valuation (3,1) differs from (3,2), so the check $↑^{\mathsf{pid}}$ fails. 
Moreover, $\mathsf{halt}$ occurs before (1,1) was observed again in combination with $\mathsf{unlock}$.
\label{fig:lock-res-words}}
\end{figure}

Using this extended storing mechanism, we can select the values of the two attributes ($↓^\mathsf{pid}$) and identify and distinguish positions in a data word where both ($↑^\mathsf{pid}$), a particular one of them ($↑^\mathsf{res}$) or a global signal (e.g., $\mathsf{halt}$) occurs.
In contrast to other decidable fragments of Freeze LTL, we are thus able to store collections of values and can compare individual values across the hierarchy of attributes.
This allows for reasoning on complex interaction of entities, also witnessed by the high, yet decidable, complexity of the logic.
\end{example}

\subparagraph{Outline and results.}
We define the semantics of LTL$^↓_{qo}$ in Section~\ref{sec:freeze} generalising Freeze LTL based on quasi-ordered attribute sets. 
We show that every fragment LTL$^↓_A$ is undecidable unless $A$ has a tree-like structure, formalised as what we call a \emph{tree-quasi-ordering}.

Section~\ref{sec:ncs} is devoted to \emph{nested counter systems (NCS)} and an analysis of their coverability problem.
We determine its non-primitive recursive complexity in terms of \emph{fast-growing complexity classes} \cite{DBLP:journals/corr/Schmitz13}.
These classes $𝐅_α$ are indexed by ordinal numbers $α$ and characterise complexities by fast-growing functions from the extended Grzegorczyk hierarchy (details are provided in Section~\ref{sec:ncs}).
We show that with increasing nesting level coverability in NCS exceeds every class $𝐅_α$ for ordinals $α<ε_0$.
By also providing a matching upper bound, we establish the following.%
\begin{theorem}[NCS]
  \label{thm:ncs-completeness}
  The coverability problem in NCS is $𝐅_{ε_0}$-complete.
\end{theorem}
We consider the fragment LTL$^↓_{tqo}$ in Section~\ref{sec:freeze-complexity}.
It restricts the available dependency relations to tree-quasi-orderings.
By reducing the satisfiability problem to NCS coverability, we obtain a precise characterisation of the decidability frontier in LTL$^↓_{qo}$.
Moreover, we transfer the lower bounds obtained for NCS to the logic setting.
This leads us to a strict hierarchy of decidable fragments of LTL$^↓_{tqo}$ parameterised by the depth of the attribute orderings and a completeness result for LTL$^↓_{tqo}$.
\begin{theorem}[LTL$^↓_{qo}$]
  \label{thm:freeze-decidability-complexity}
  The satisfiability problem of
  \begin{itemize}
    \item LTL$^↓_A$ is decidable if and only if $A$ is a tree-quasi-ordering.
    \item LTL$^↓_{tqo}$ is $𝐅_{ε_0}$-complete.
  \end{itemize}
\end{theorem}

\subparagraph{Related work.}

The \emph{freeze} \cite{DBLP:conf/podc/Henzinger90} mechanism was introduced as a natural form of storing and comparing (real-time) data at different positions in time \cite{DBLP:journals/jacm/AlurH94} and since studied extensively in different contexts, e.g., 
\cite{DBLP:journals/jolli/Goranko96,DBLP:journals/logcom/Fitting02,DBLP:conf/time/LisitsaP05}.
In particular linear temporal logic employing the freeze mechanism over domains with only equality, i.e., data words, was considered in \cite{DBLP:conf/time/DemriLN05} and shown highly undecidable ($Σ^1_1$-hard).
Therefore, several decidable fragments were proposed in the literature with complexities ranging from exponential \cite{DBLP:conf/fsttcs/Lazic06} and double-exponential space \cite{DBLP:conf/lics/DemriFP13} to non-primitive recursive complexities \cite{DBLP:journals/tocl/DemriL09}. 
For the one-register fragment LTL$^↓_1$ that we build on here, an $𝐅_ω$ upper bound was given in \cite{DBLP:journals/corr/abs-1202-3957}.
Due to its decidability and expressiveness, it is called in~\cite{DBLP:journals/tocl/DemriL09} a “competitor” for the two-variable first-order logic over data words FO$^2(∼,<,+1)$ studied in \cite{DBLP:journals/tocl/BojanczykDMSS11}. 
There, satisfiability was reduced to and from reachability in Petri nets in double-exponential time and polynomial time, respectively, for which recent results provide an $𝐅_{ω^3}$ upper bound \cite{DBLP:conf/lics/LerouxS15}.

Our main ambition is to incorporate means of storing and comparing collections of data values.
The apparent extension of storing and comparing even only pairs generically renders logics on data words, even those with essential restrictions, undecidable \cite{DBLP:journals/tocl/BojanczykDMSS11,DBLP:conf/fsttcs/KaraSZ10,DBLP:conf/concur/DeckerHLT14}.
This applies in particular to fragments of LTL$^↓_1$~\cite{DBLP:conf/lics/DemriFP13}.

The logic \emph{Nested Data LTL (ND-LTL)} studied in \cite{DBLP:conf/concur/DeckerHLT14} employs a navigation concept based on an \emph{ordered} set of attributes.
It inspired our extension of Freeze LTL but in contrast, data values in ND-LTL are not handled explicitly, resulting in incomparable expressiveness and different notions of natural restrictions.
While ND-LTL also features a freeze-like mechanism, it does not contain an explicit check operator ($↑$).
Instead, data-aware variants of temporal operators such as $\U^=$ express constraints (only) for position where the stored value is present.
For example, an ND-LTL formula
$\G(\mathsf{lock} → ↓^{\mathsf{pid}} (¬\mathsf{halt} \U^= \mathsf{unlock}))$
(in notation of this paper) requires that for every position satisfying $\mathsf{lock}$ there is a future position $\mathsf{unlock}$ with the same data value and that $¬\mathsf{halt}$ holds (at least) on all those position in between that also carry this particular value.
In contrast, 
$\G(\mathsf{lock} → ↓^{\mathsf{pid}} (¬\mathsf{halt} \U (\mathsf{unlock}\,∧↑^{\mathsf{pid}})))$ asserts that at \emph{all} position in between $¬\mathsf{halt}$ holds.
Enforcing such constraints in ND-LTL typically requires an additional level of auxiliary attributes.

The future fragment ND-LTL$^+$ was shown decidable and non-primitive recursive on finite $A$-attributed data words for tree-(partial-)ordered attribute sets $A$.
However, no upper complexity bounds were provided and the developments in this paper significantly raise the lower bounds (cf.\ Section~\ref{sec:conclusion}).
The influence of more general attribute orderings, in particular the precise decidability frontier in that dimension, was not investigated for ND-LTL and its fragments.
Instead, the logic was shown undecidable by exploiting the combination of future- and past-time operators.
Extending LTL$^↓_1$ with past-time operators is also known to lead to undecidability.
ND-LTL$^+$ stays decidable even on infinite words, which is not the case for LTL$^↓_{tqo}$ since satisfiability of LTL$^↓_1$ is already $Π^0_1$-complete \cite{DBLP:journals/tocl/DemriL09}.

%% file: sec_freeze.tex
\section{Semantics and Undecidability of LTL$^↓_{qo}$}
\label{sec:freeze}

By specifying dependencies between attributes from a set $A$ in terms of a quasi-ordering ${⊑}⊆A×A$ the freeze mechanism can be used to store the values of multiple attributes at once.
The essential intuition for our generalised storing mechanism can well be obtained from the special case of a linearly ordered attribute set $[k]=｛1,…,k｝$ with the natural ordering $≤$ for some natural number $k∈ℕ$.
In fact, many of the technical developments in this paper concerning decidability and complexity are carried out within this setting but for concise presentation we only provide the most general formulation that captures also the undecidable case.

The valuations $𝐝∈Δ^{[k]}$ in a $[k]$-attributed data word are essentially sequences (or vectors) $d_1…d_k$ where the $x$-th position carries the value $d_x=𝐝(x)$ of attribute $x∈[k]$.
Note that Example~\ref{ex:lock-res} matches that setting when renaming $\mathsf{res}$ to 1 and $\mathsf{pid}$ to 2.

In a formula $↓^xφ$ the subformula $φ$ is evaluated in the context of the current value $𝐝(x)$ of the attribute $x∈A$ and the values $𝐝(y)$ of all smaller attributes $y≤x$.
Thus, the \emph{prefix} $d_1…d_x$ of the value sequence at the current position is stored for later comparison.
A check operator $↑^y$ then compares the stored values $d_1…d_x$ to the values $d_1'…d_y'$ at the current position: the check is successful if the latter sequence is a prefix of the former, i.e.\ $y≤x$ and $d_1…d_y = d_1'…d_y'$.

For the general setting of arbitrary (quasi-ordered) dependency relations $(A,⊑)$, we lift the notion of the prefix of length $x$ to the \emph{restriction} of a valuation $𝐝:A→Δ$ to the \emph{downward-closure} $\cl(x)=｛y∈A｜y⊑x｝$ of $x$ in $A$.
This restriction is defined as $𝐝|_x:\cl(x) → Δ$ with $𝐝|_x(y) = 𝐝(y)$ for $y⊑x$. 
We denote the set of all such partial data valuations by $Δ_⊥^A = ｛𝐝: \cl(x) → Δ ｜x∈A｝$.

Partial valuations $𝐝,𝐝'∈Δ_⊥^A$ are compared in analogy to sequences: it must be possible to map one onto the other such that ordering is preserved and all values coincide.
Formally, we define an equivalence relation ${≃}⊆Δ_⊥^A × Δ_⊥^A$ by $𝐝≃𝐝'$ if and only if there is a bijection $h:\dom(𝐝) → \dom(𝐝')$ such that, for all attributes $x∈\dom(𝐝)$ we have $𝐝(x) = 𝐝'(h(x))$ and, for all attributes $y∈\dom(𝐝)$, $x⊑y ⇔ h(x)⊑h(y)$.
Notice that this requires the domains of $𝐝$ and $𝐝'$ to be isomorphic.
In the definition presented next we therefore allow for restricting the stored valuation arbitrarily before it is matched against the current one.
In the linear case this simply means truncating the stored sequence before comparison and intuitively it allows for removing unnecessary information from the context.

\subparagraph{Semantics of LTL$^↓_{qo}$.} 

For a non-empty data word $w=(a_1,𝐝_1)…(a_n,𝐝_n)∈(Σ×Δ^A)^+$, an index $1≤i≤n$ in $w$ and a partial data valuation $𝐝∈Δ_⊥^A$ the semantics of LTL$^↓_A$ formulae is defined inductively by
\[\begin{array}{ll@{\quad}c@{\quad}l}
  (w,i,𝐝) &⊧ a_i   \\
  (w,i,𝐝) &⊧ ¬φ    &:⇔ &(w,i,𝐝)\not⊧φ\\
  (w,i,𝐝) &⊧ φ ∧ ψ &:⇔ &(w,i,𝐝)⊧φ \text{ and } (w,i,𝐝)⊧ψ\\
  (w,i,𝐝) &⊧ \Xφ   &:⇔ &i+1≤n \text{ and } (w,i+1,𝐝)⊧φ\\
  (w,i,𝐝) &⊧ φ\Uψ  &:⇔ &∃_{i≤k≤n}: (w,k,𝐝)⊧ψ \text{ and } ∀_{i≤j<k}: (w,j,𝐝)⊧φ\\
  (w,i,𝐝) &⊧ ↓^xφ  &:⇔ &(w,i,𝐝_i|_x)⊧φ\\
  (w,i,𝐝) &⊧ ↑^x   &:⇔ & ∃_{y∈A} : 𝐝|_y ≃ 𝐝_i|_x.
\end{array}\]
For formulae $φ$ where every check operator $↑^x$ is within the scope of some freeze quantifier $↓^y$ the stored valuation is irrelevant and we write $w⊧φ$ if $(w,1,𝐝)⊧ φ$ for any valuation $𝐝$.

\begin{example}
Consider a set of attributes $\mathsf{A}=｛\mathsf{x_1},\mathsf{x_2},\mathsf{x_3},\mathsf{y_1},\mathsf{y_2}｝$ with $\mathsf{x_1}⊑\mathsf{x_2}⊑\mathsf{x_3}$ and $\mathsf{y_1}⊑\mathsf{y_2}$ (this is an example for a tree-quasi-ordering, see below), the formula $↓^\mathsf{x_3} \X(↑^{\mathsf{y_2}}\U↑^\mathsf{x_3})$ and a data word $w=(a_1,𝐝_1)…(a_n,𝐝_n)$.
The formula reads as: “Store the current values $d_1,d_2,d_3$ of $\mathsf{x_1},\mathsf{x_2},\mathsf{x_3}$, respectively. Move on to the next position. Verify that the stored value $d_1$ appears in $\mathsf{y_1}$ and that $d_2$ appears in $\mathsf{y_2}$ until the values $d_1,d_2,d_3$ appear again in attributes $\mathsf{x_1},\mathsf{x_2},\mathsf{x_3}$, respectively.”

At the first position, the values $d_1=𝐝_1(\mathsf{x_1})$, $d_2=𝐝_1(\mathsf{x_2})$ and $d_3 = 𝐝_1(\mathsf{x_3})$ are stored in terms of the valuation $𝐝=𝐝_1|_\mathsf{x_3}: ｛\mathsf{x_1},\mathsf{x_2},\mathsf{x_3}｝ → Δ$ since $\mathsf{x_1},\mathsf{x_2},\mathsf{x_3}$ depend on $\mathsf{x_3}$.
Assume for the second position $𝐝_2(\mathsf{x_1})≠𝐝_1(\mathsf{x_1})=d_1$.
The formula $↑^\mathsf{x_3}$ is not satisfied at the second position in the context of $𝐝$ since the only attribute $p∈A$
such that $\cl(p)$ is isomorphic to $｛\mathsf{x_1},\mathsf{x_2},\mathsf{x_3}｝$ is $p=\mathsf{x_3}$.
Then, however, any order preserving isomorphism needs to map $\mathsf{x_1}∈\dom(𝐝)$ to $\mathsf{x_1}∈\dom(𝐝_2)$ since $\mathsf{x_1}$ is the minimal element in both domains but $𝐝(\mathsf{x_1})≠𝐝_2(\mathsf{x_1})$.
The only way to not violate the formula is hence that $𝐝_2(\mathsf{y_1})=𝐝_1(\mathsf{x_1})$ and $𝐝_2(\mathsf{y_2})=𝐝_1(\mathsf{x_2})$.
Then, we can choose $p=\mathsf{x_2}$ and have $𝐝|_{\mathsf{x_2}}≃𝐝_2|_{\mathsf{y_2}}$ meaning that $↑^{\mathsf{y_2}}$ is satisfied.
\end{example}

%% file: sec_freeze_undecidability.tex
\subparagraph{Undecidability.}

For ${⊑}=｛(x,x)｜x∈A｝$ (identity) we obtain the special case where only single values can be stored and compared. 
If $|A|=1$ we obtain the one-register fragment LTL$^↓_1$.
On the other hand, if $A$ contains three attributes $x,y,z$ such that $x$ and $y$ are incomparable and $x⊑z⊒y$ then storing the value of $z$ also stores the values of $x$ and $y$.
This amounts to storing and comparing the \emph{set} $｛d_x,d_y｝⊂Δ$ of values assigned to $x$ and $y$.
This is not precisely the same as storing the ordered \emph{tuple} $(d_x,d_y)∈Δ×Δ$ but together with the ability of storing and comparing $x$ and $y$ independently it turns out to be just as contagious considering decidability.

In \cite{DBLP:journals/tocl/BojanczykDMSS11} it is shown that the satisfiability problem of two-variable first-order logic over data words with two class relations is undecidable by reduction from \emph{Post's correspondence problem}.
We can adapt this proof and formulate the necessary conditions for a data word to encode a solution using only the attributes $x⊑z⊒y$. 
With ideas from \cite{DBLP:conf/lics/DemriFP13} we can also omit using past-time operators.
Moreover, this result can be generalised to arbitrary quasi-orderings that contain three attributes $x⊑z⊒y$.

The absence of such a constellation is formalised by the notion of a \emph{tree-quasi-ordering} defined as a quasi-ordering where the downward-closure of every element is totally ordered.
This precisely prohibits elements $z$ that depend on two independent elements $x$ and $y$.
The definition describes in a general way a hierarchical, tree-like structure.
Intuitively, a tree-quasi-ordering is (the reflexive and transitive closure of) a forest of strongly connected components.

\begin{restatable}[Undecidability]{theorem}{thmundecidability}\label{thm:undecidability}
  Let $(A,⊑)$ be a quasi-ordered set of attributes that is \emph{not} a tree-quasi-ordering. 
  Then the satisfiability problem of LTL$^↓_A$ is $Σ^0_1$-complete over $A$-attributed data words.
\end{restatable}

See Appendix~\ref{sec:apx-freeze-undecidability} for a complete proof. 
As will be discussed in Section~\ref{sec:freeze-complexity}, tree-quasi-orderings represent not only necessary but also sufficient conditions for the logic to be decidable.

%% file: sec_ncs.tex
\section{Nested Counter Systems}\label{sec:ncs}

\emph{Nested counter systems (NCS)} are a generalisation of counter systems similar to higher-order multi-counter automata as used in~\cite{DBLP:conf/mfcs/BjorklundB07} and nested Petri nets~\cite{DBLP:conf/ershov/LomazovaS99}.
In this section we establish novel complexity results for their coverability problem.
A finite number of counters can equivalently be seen as a multiset $M=｛c_1:n_1,…,c_m:n_m｝$ over a set of counter names $C=｛c_1,…,c_n｝$.
We therefore define NCS in the flavor of~\cite{DBLP:conf/concur/DeckerHLT14} as systems transforming \emph{nested multisets}.

Let $𝔐(A)$ denote, for any set $A$, the set of all \emph{finite multisets} of elements of $A$. 
For $k∈ℕ$ we write $[k]$ to denote the set $｛1,…,k｝⊂ℕ$ with the natural linear ordering $≤$.
A \emph{$k$-nested counter system ($k$-NCS)} is a tuple $𝓝 = (Q,δ)$ comprised of a finite set $Q$ of \emph{states} and a set $δ⊆⋃_{i,j∈[k+1]}(Q^i×Q^j)$ of \emph{transition rules}.
For $0≤i≤k$ the set $𝓒_i$ of \emph{configurations of level $i$} is inductively defined by $𝓒_k = Q$ and $𝓒_{i-1} = Q × 𝔐(𝓒_i)$.
The set of \emph{configurations} of $𝓝$ is then $𝓒_𝓝 = 𝓒_0$.
Every element of $C_𝓝$ can, more conveniently, be presented as a term constructed over unary function symbols $Q$, constants $Q$ and a binary operator $+$ that is associative and commutative.
For example, the configuration 
$(q_0, ｛(q_1,∅):1, 
         (q_1, ｛(q_2,∅): 2｝):2,
         (q_1, ｛(q_2,∅): 2, 
                 (q_3,｛(q_4,∅):1｝):1
                ｝):1
        ｝
)$
can be represented by the term
$q_0(
    q_1 
  + q_1(q_2 + q_2) 
  + q_1(q_2 + q_2) 
  + q_1(q_2 + q_2 + q_3(q_4))
)$.
The operational semantics of $𝓝$ is now defined in terms of the \emph{transition relation} ${→}⊆𝓒_𝓝×𝓒_𝓝$ on configurations given by rewrite rules.
For $((q_0,…,q_i), (q_0',…,q_j')) ∈δ$ and $i,j<k$ we let  
\begin{align*}
q_0(X_1+q_1(… q_i(X_{i+1})…)) &→ q'_0(X_1+q_1'(… q'_{j}(X_{j+1})…))%
\end{align*}
for any $X_h∈𝔐(C_h)$ where $1≤h≤k$ and $X_ℓ=∅$ for $i+2≤ℓ≤j+1$.
For example, a rule $((q_0), (q_0'))$ changes the state $q_0$ in the example configuration above to $q_0'$. 
A rule $((q_0,q_1),(q_0,q_1,q_2'))$ adds a state $q_2'$ non-deterministically as a direct child of one of the states $q_1$
resulting in one of the three configurations
\begin{align*}
&q_0(
    q_1(q_2') 
  + q_1(q_2 + q_2) 
  + q_1(q_2 + q_2) 
  + q_1(q_2 + q_2 + q_3(q_4))
),\\
&q_0(
    q_1 
  + q_1(q_2 + q_2 + q_2') 
  + q_1(q_2 + q_2) 
  + q_1(q_2 + q_2 + q_3(q_4))
) \text{ and } \\
&q_0(
    q_1 
  + q_1(q_2 + q_2) 
  + q_1(q_2 + q_2) 
  + q_1(q_2 + q_2 + q_3(q_4) + q_2')
).
\end{align*}
Moreover, a rule $((q_0,q_1,q_3),(q_0))$ would remove specifically and completely the sub-configuration $q_1(q_2 + q_2 + q_3(q_4))$ since it does not match any other one.

The remaining cases for transitions, where (\ref{eqn:j-eq-k}) $i<k=j$, (\ref{eqn:j-le-k-i-eq-k}) $i=k>j$ and (\ref{eqn:j-eq-k-i-eq-k}) $i=k=j$, are defined as expected by rules
\begin{align}
q_0(X_1+q_1(… q_i(X_{i+1})…)) &→ q'_0(X_1+q_1'(… q'_{k-1}(X_{k}+q_k')…))\label{eqn:j-eq-k}\\
q_0(X_1+q_1(… q_{k-1}(X_k+q_k)…)) &→ q'_0(X_1+q_1'(… q'_{j}(X_{j+1})…))\label{eqn:j-le-k-i-eq-k}\\
q_0(X_1+q_1(… q_{k-1}(X_k+q_k)…)) &→ q'_0(X_1+q_1'(… q'_{k-1}(X_{k}+q_k')…))\label{eqn:j-eq-k-i-eq-k}
\end{align}
respectively, where for (\ref{eqn:j-eq-k}) we have $X_{i+2}=…=X_k=∅$.
Note that these cases are exhaustive since the nesting depth of terms representing configurations from $𝓒_𝓝$ is at most $k$.
As usual we denote by $→^*$ the reflexive and transitive closure of $→$.
By $⪯$ we denote the nested multiset ordering, i.e.\ $M' ⪯ M$ iff $M'$ can be obtained by removing elements (or nested multisets) from $M$.
Given two configurations $C,C'∈𝓒_𝓝$ the \emph{coverability problem} asks for the existence of a configuration $C''∈𝓒_𝓝$ with $C''⪰C'$ and $C→^*C''$.

To establish our complexity results on NCS we require some notions on ordinal numbers, ordinal recursive functions and respective complexity classes.
We represent ordinals using the \emph{Cantor normal form (CNF)}.
An ordinal $α<ε_0$ is represented in CNF as a term $α = ω^{α_1} + … + ω^{α_k}$ over the symbol $ω$ and the associative binary operator $+$ where $α > α_1 ≥ … ≥ α_k$.
Furthermore, we denote \emph{limit ordinals} by $λ$.
These are ordinals such that $α+1<λ$ for every $α<λ$. 
We associate them with a \emph{fundamental sequence} $(λ_n)_n$ with supremum $λ$ defined by
\[
  (α + ω^{β+1})_n := α + \overbrace{ω^β + … + ω^β}^n \qquad \text{ and } \qquad (α+ω^{λ'})_n := α+ω^{λ'_n}
\]
for ordinals $β$ and limit ordinals $λ'$.
Then, $ε_0$ is the smallest ordinal $α$ such that $α = ω^α$.
We denote the $n$-th exponentiation of $ω$ as $Ω_n$, i.e. $Ω_1 := ω$ and $Ω_{n+1} := ω^{Ω_n}$.
Consequently, $(Ω_n)_m < Ω_n$ is the $m$-th element of the fundamental sequence of $Ω_n$.
Given a monotone and expansive\footnote{A function $f: A → A$ over an ordering $(A,≤)$ is monotone if $a≤a' ⇒ f(a)≤f(a')$ and expansive if $a≤f(a)$ for all $a,a'∈A$.} function $h: ℕ→ℕ$, a \emph{Hardy hierarchy} is an ordinal-indexed family of functions $h^α: ℕ → ℕ$ defined by
$h^0(n) := n$, $h^{α+1}(n) := h^α(h(n))$ and $h^{λ}(n) := h^{λ_n}(n)$.
Choosing $h$ as the incrementing function $H(n) := n + 1$, the \emph{fast growing hierarchy} is the family of functions $F_α(n)$ with $F_α(n) := H^{ω^α}(n).$

The hierarchy of \emph{fast growing complexity classes} $𝐅_α$ for ordinals $α$ is defined in terms of the fast-growing functions $F_α$. 
We refer the reader to \cite{DBLP:journals/corr/Schmitz13} for details
and only remark that $𝐅_{<ω}$ is the class of primitive recursive problems and problems in $𝐅_ω, 𝐅_{ω^ω}$ are solvable with resources bound by Ackermannian and Hyper-Ackermannian functions, respectively.
The fact most relevant for our classification is that a basic $𝐅_α$-complete problem is the termination problem of Minsky machines $𝓜$ where the sum of the counters is bounded by $F_α(|𝓜|)$ \cite{DBLP:journals/corr/Schmitz13}.

\subparagraph{Upper bound.}
To obtain an upper bound for the coverability problem in $k$-NCS we reduce it to that in \emph{priority channel systems (PCS)}~\cite{DBLP:journals/corr/abs-1301-5500}. 
PCS are comprised of a finite control and a fixed number of channels, each storing a string to which a letter can be appended (write) and from which the first letter can be read and removed (read).
Every letter carries a priority and can be lost at any time and any position in a channel if its successor in the channel carries a higher or equal priority.
PCS can easily simulate NCS by storing and manipulating an NCS configuration in a channel where a state $q$ at level $i>0$ in the NCS configuration is encoded by a letter $(q,k-i)$ with priority $k-i$.
E.g., the 3-NCS configuration 
$ q_0(q_1 
  + q_1(q_2 + q_2) 
  + q_1(q_2 + q_2 + q_3(q_4)))
$
can be encoded as a channel of the form
$(q_1,2)(q_1,2) (q_2,1) (q_2,1) (q_1,2)  (q_2,1) (q_2,1) (q_3,1) (q_4,0)$
while $q_0$ is encoded in the finite control.

Taking the highest priority for the outermost level ensures that the lossiness of PCS corresponds to descending with respect to $⪯$ for the encoded NCS configuration.
Thus the coverability problem in NCS directly translates to that in PCS.
The coverability (control-state reachability) problem in PCS with one channel and $k$ priorities lies in the class $𝐅_{Ω_{2k}}$~\cite{DBLP:journals/corr/abs-1301-5500} and we thus obtain an upper bound for NCS coverability.
See Appendix~\ref{apx:ncs-upper-bound} for further details.

\begin{restatable}{proposition}{prpncstopcs}\label{prp:ncstopcs}
Coverability in $k$-NCS is in $𝐅_{Ω_{2k}}$.
\end{restatable}

%% file: sec_ncs_lowerbound.tex
\subparagraph{Lower bound.}
We can reduce, for any $k>1$, the halting problem of $H^{(Ω_k)_l}$-bounded Minsky machines to coverability in $k$-NCS with the number of states bounded by $l+c$, for some constant $c$.
This yields the following characterisation (recall that $H^{(Ω_{k+1})_l}=F_{(Ω_k)_l}$).
\begin{theorem}\label{thm:ncshardness}
Coverability in $(k+1)$-NCS is $𝐅_{Ω_{k}}$-hard.
\end{theorem}
The idea is to construct a $k$-NCS $𝓝=(Q,δ)$ that can simulate the evaluation of the Hardy function $H^α(n)$ for $α≤(Ω_k)_l$ in forward as well as backward direction.
It can then compute a \emph{budget} that is used for simulating the Minsky machine.
Lower bounds for various models were obtained using this scheme for Turing machines \cite{DBLP:conf/lics/ChambartS08,DBLP:journals/corr/abs-1301-5500} or Minsky machines \cite{DBLP:conf/mfcs/Schnoebelen10,RosaVelardo14}.

The following construction uses $k+1$ levels of which one can be eliminated later.
We encode the ordinal parameter $α$ of $H^α(n)$ and its argument $n∈ℕ$ (unary) into a configuration
\vspace{-0.3cm}
\[
  C_{α, n} := \mathsf{main}(\mathsf{s}(M_α) + \mathsf{c}(\overbrace{1 + … + 1}^n))
\]
using control-states $\mathsf{main},\mathsf{s},\mathsf{c},ω∈Q$ and configurations $M_α$ defined by $M_{0} := ∅$ and $M_{ω^α + β} := ω(M_α) + M_β$.
For example, an ordinal $α=ω^ω+ω^2+ω^2+1$ is encoded by 
\begin{align*}
  M_α =｛(ω,｛(ω,｛(ω,∅):1｝):1｝):1, (ω,｛(ω,∅):2｝):2, (ω,∅):1｝
\end{align*}
Note that we use shorthands for readability, e.g., $ω^ω$ stands for $ω^{ω^1}$ where $1$ is again short for the ordinal $ω^0$.
The construction has to fulfil the following two properties.
As NCS do not feature a zero test exact simulation cannot be enforced but errors can be restricted to be ``lossy''.
\begin{lemma}
\label{lem:ncshardness1}
For all configurations $C_{α,n} →^* C_{α',n'}$ we have $H^α(n) ≥ H^{α'}(n')$.
\end{lemma}
The construction will, however, admit at least one run maintaining exact values.
\begin{lemma}
\label{lem:ncshardness2}
If
$H^α(n) = H^{α'}(n')$
then there is a run
$C_{α,n} →^* C_{α',n'}$.
\end{lemma}
The main challenge is simulating a computation step from a limit ordinal to an element of its fundamental sequence, i.e., from $C_{α + λ,n}$ to $C_{α+λ_n,n}$ and conversely.
Encoding the ordinal parameter using multisets loses the ordering of the addends of the respective CNF terms.
Thus, instead of taking the last element of the CNF term we have to select the smallest element, with respect to $⪯$, of the corresponding multiset.
To achieve that, we extend NCS by two operations \cp and \min.
Given some configuration $C = q_1(q_2(M))∈𝓒_{𝓝}$ the operation $(q_1, q_2) \cp (q_1', q_2')$ copies $M$ resulting in $C' = q_1'(q_2'(M_1) + q_2'(M_2))$ with $M_1, M_2 ⪯ M$.
Conversely, given the configuration $C'$ the operation $(q_1', q_2') \min (q_1, q_2)$ results in $C$ with $M ⪯ M_1, M_2$.

Both operations can be implemented in a depth first search fashion using the standard NCS operations.
Based on them the selection of a smallest element from a multiset can be simulated:
all elements are copied (non-deterministically) one by one to an auxiliary set while enforcing a descending order. 
Applying the \min operation in every step ensures that we either proceed indeed in descending order or make a ``lossy'' error.
We guess, in each step, whether the smallest element is reached and in that case delete the source multiset.
Thereby it is ensured, that the smallest element has been selected or, again, a ``lossy'' error occurs such that the selected element is now the smallest one.
The additional level in the encoding of $C_{α,n}$ enables us to perform this deletion step.

A similar idea to select a smallest ordinal from a multiset is used in \cite{RosaVelardo14}.
However, we need to handle nested structures of variable size correctly whereas in this work the considered ordinals are below $Ω_3$.
They are represented by a multiset of vectors of fixed length where the vectors can be compared and modified directly in order enforce the choice of a minimal one.

We now construct an NCS simulating an $H^α(s)$-bounded Minsky machine $𝓜$ of size $s:=|𝓜|$ analogously to the constructions in \cite{DBLP:conf/mfcs/Schnoebelen10,DBLP:conf/lics/ChambartS08,DBLP:journals/corr/abs-1301-5500}.
It starts in a configuration $C_{α, s}$ to evaluate $H^{α}(s)$.
When it reaches $C_{0, n}$ for some $n≤H^α(s)$ it switches its control state and starts to simulate $𝓜$ using $n$ as a budget for the sum of the two simulated counters.
Zero tests can then be simulated by resets (deleting and creating multisets) causing a ``lossy'' error in case of an actually non-zero counter.
When the simulation of $𝓜$ reaches a final state the NCS moves the current counter values back to the budged counter and performs a construction similar to the one above but now evaluating $H^{α}(s)$ \emph{backwards} until reaching $(C_{α, s})'$, the initial configuration with a different control state.
If $(C_{α, s})'$ can be reached (or even covered) no ``lossy'' errors occurred and the Minsky machine $𝓜$ was thus simulated correctly regarding zero tests.
The detailed construction is presented in Appendix~\ref{apx:ncshardness}.

%% file: sec_freeze_decidability.tex
\section{From LTL$^↓_{tqo}$ to NCS and Back}
\label{sec:freeze-complexity}

Theorem~\ref{thm:undecidability} established a necessary condition for $LTL^↓_A$ to have a decidable satisfiability problem, namely that $A$ is a tree-quasi-ordering.
In the following we show that this is also sufficient.
Let LTL$^↓_{tqo}$ denote the fragment of LTL$^↓_{qo}$ restricted to tree-quasi-ordered sets of attributes.
The decidability and complexity results for NCS can be transferred to  LTL$^↓_{tqo}$ to obtain upper and lower bounds for the satisfiability problem of the logic.

We show a correspondence between the nesting depth in NCS and the depths of the tree-quasi-ordered attribute sets that thus constitutes a semantic hierarchy of logical fragments.
We provide the essential ideas in the following and refer the reader to Appendix~\ref{apx:freeze2ncs} and~\ref{apx:ncs2freeze} for the detailed constructions.

The \emph{depth} of a finite tree-quasi-ordering $A$ is the maximal length $k$ of strictly increasing sequences $x_1⊏x_2⊏…⊏x_k$ of attributes in $A$.
The first observation is that we can reduce satisfiability of any LTL$^↓_{tqo}$ formula over attributes $A$ to satisfiability of an LTL$^↓_{[k]}$ formula where $[k]=｛1,…,k｝$ is an initial segment of the natural numbers with natural \emph{linear ordering} and $k$ is the depth of $A$.

\begin{restatable}[LTL$^↓_{tqo}$ to LTL$^↓_{[k]}$]{proposition}{prplinearisation}
\label{prp:linearisation}
For a tree-quasi-ordered attribute set $A$ of depth $k$ every $\LTL^↓_A$ formula can be translated to an equisatisfiable $\LTL^↓_{[k]}$ formula of exponential size.
\end{restatable}

To reduce an arbitrary tree-quasi-ordering $A$ of depth $k$ we first remove maximal strongly connected components (SCC) in the graph of $A$ and replace each of them by a single attribute.
This does only affect the semantics of formulae $φ$ if attributes are compared that did not have an isomorphic downward-closure in $A$.
These cases can, however, be handled by additional constraints added to $φ$.
Data words over a thus obtained \emph{partially} ordered attribute set of depth $k$ can now be encoded into words over the \emph{linear} ordering $[k]$ of equal depth $k$.
The idea is to encode a single position into a frame of positions in the fashion of \cite{DBLP:conf/fsttcs/KaraSZ10,DBLP:conf/concur/DeckerHLT14}.
That way a single attribute on every level suffices.
Any formula can be transformed to operate on these frames instead of single positions at the cost of an at most exponential blow-up.

\subparagraph{From LTL$^↓_{[k]}$ to NCS.}

Given an LTL$^↓_{[k]}$ formula $Φ$ we can now construct a $(k+1)$-NCS $𝓝$ and two configurations $C_{init},C_{final}∈𝓒_𝓝$ such that $Φ$ is satisfiable if and only if $C_{final}$ can be covered from $C_{init}$.

The idea is to encode sets of \emph{guarantees} into NCS configurations.
These guarantees are subformulae of $Φ$ and are guaranteed to be satisfiable.
The constructed NCS can instantiate new guarantees and combine existing ones while maintaining the invariant that there is always a data word $w∈(Σ×Δ^{[k]})^+$ that satisfies all of them.
To ensure the invariant, the guarantees are organised in a forest of depth $k$ as depicted in Figure~\ref{fig:threads}.

\begin{figure}[t]
\begin{center}
  \input{fig_threads1}
\end{center}
\caption{Example of a guarantee forest of depth 3 maintained and modified by the NCS constructed for some LTL$^↓_{[3]}$ formula. 
Node enumeration (grey) is only for reference.
\label{fig:threads}}
\end{figure}
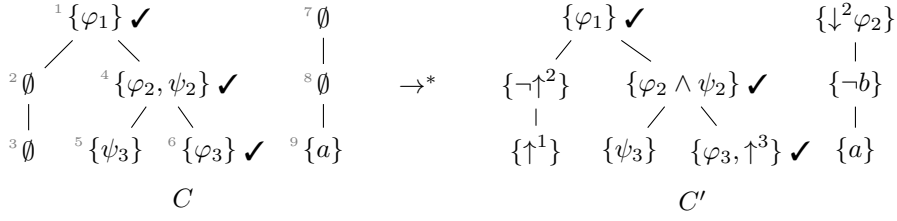

All formulae $φ$ contained in the same node $v$ of this forest are moreover not only satisfied by the same word $w$ but also with respect to a common valuation $𝐝_v∈Δ^{[k]}_⊥$, i.e., $(w,1,𝐝_v) ⊧ φ$.
Recall that valuations over linearly ordered attributes can be seen as sequences.
The forest structure now represents the common-prefix relation between these valuations $𝐝_v$.
For two nodes $v$, $v'$ having a common ancestor at level $i∈[k]$ in the forest, the corresponding valuations $𝐝_v$, $𝐝_{v'}$ can be chosen such that they agree on attributes $1$ to $i$.
A uniquely marked branch in the forest further represents the valuation $𝐝_1$ at the first position in $w$.
If a formula $φ$ is contained in the marked node at level $i$ in the forest then $(w,1,𝐝_1|_i)⊧ φ$. 
In that case $(w,1,𝐝)⊧ ↓^iφ$ holds for any $𝐝∈Δ^{[k]}_⊥$ and the formula $↓^iφ$ could be added to any of the nodes in the forest without violating the invariant.
Similarly, for a marked node $v$ at level $i$ the formulae $↑^i$ can be added to any node in the subtree with root $v$.
Moreover, other atomic formulae, Boolean combinations, and temporal operators can also be added consistently.
The NCS $𝓝$ can perform such modifications on the forest, represented by its configuration, by corresponding transitions.
\begin{example}
Consider the two guarantee forests depicted in Figure~\ref{fig:threads} that are encoded in configurations $C$ and $C'$ of an NCS constructed for some LTL$^↓_{[3]}$ formula.
The invariant is the existence of a word $w=(a,𝐝)…$ and valuations $𝐝_v∈Δ^{[i]}$ such that $(w,1,𝐝_v)$ satisfies the formulae in a node $v$ at level $i$.
The forest structure relates these valuations to $𝐝$ (nodes marked by $\cmark$) and each other.
E.g., $(w,1,𝐝|_1) ⊧φ_1$, $(w,1,𝐝|_2) ⊧φ_2$ and there is $𝐞$ with $𝐞|_2=𝐝|_2$ and $𝐞(3)≠𝐝(3)$ s.t.\ $(w,1,𝐞)⊧ψ_3$.
Let $v_1,…,v_9$ be the nodes of the forest (as enumerated in the figure).
Several possible operations are exemplified by the transition between $C$ and $C'$.
The formula $↑^3$ can be added to the node $v_6$ containing the formula $φ_3$ since that node is checked on level 3.
Similarly, there is $𝐝_3$ for node $v_3$ such that $𝐝_3(1) = 𝐝(1)$ and hence $(w,𝐝_3)⊧↑^1$.
The formula $↑^2$ cannot be added to the node $v_2$ since it is not below the checked node on level two. 
Consequently, the node can contain $¬↑^2$.
Node $v_4$ on level 2 does already contain $φ_2$ and $ψ_2$, meaning they are both satisfied by $w$ and a valuation $𝐝_4∈Δ^{[2]}$.
Hence the same holds for their conjunction.
Moreover, $v_4$ is checked and therefore $𝐝_4 = 𝐝|_2$.
This implies that $(w,𝐝')⊧↓^2φ_2$ for any $𝐝'$ and that the formula can be added to any node in the tree, e.g.\ $v_7$.
\end{example}
Recall that we only need to consider subformulae of $Φ$ and thus remain finite-state for representing nodes.
More precisely, the number of states in $𝓝$ is exponential in the size of $Φ$ since they encode sets of formulae.

A crucial aspect is how the NCS can consistently add formulae of the form $\Xφ$.
This needs to be done for all stored guarantees at once but NCS do not have an atomic operation for modifying all states in a configuration.
Therefore, the forest is copied recursively, processing each copied node.
The NCS $𝓝$ can choose at any time to stop and remove the remaining nodes.
That way it might loose guarantees but maintains the invariant since only processed nodes remain in the configuration.
The forest of depth $k$ itself could be maintained by a $k$-NCS but to implement the copy operation an additional level is needed.

The initial configuration $C_{init}$ consists of a forest without any guarantees.
In a setup phase, the NCS can add branches and formulae of the form $\WXφ$ since they are all satisfied by any word of length 1.
Once the formula $Φ$ is encountered in the current forest the NCS can enter a specific target state $q_{final}$. 
A path starting in $C_{init}$ and covering the configuration $C_{final}=q_{final}$ then constitutes a model of $Φ$ and vice versa.

\begin{restatable}{theorem}{thmfreezetoncs}
\label{thm:freezetoncs}
  For tree-quasi-ordered attribute sets $A$ with depth $k$ satisfiability of LTL$^↓_A$ can be reduced in exponential space to coverability in $(k+1)$-NCS.
\end{restatable}

\subparagraph{From $k$-NCS to LTL$^↓_{[k]}.$}

 Let $𝓝=(Q,δ)$ be a $k$-NCS.
We are interested in describing witnesses for coverability. 
It suffices to construct a formula $Φ_𝓝$ that characterises precisely those words that encode a \emph{lossy} run from a configuration $C_{start}$ to a configuration $C_{end}$.
We call a sequence $C_0C_1…C_n$ of configurations $C_j∈𝓒_𝓝$ a lossy run from $C_0$ to $C_n$ if there is a sequence of intermediate configurations $C_0'…C_{n-1}'$ such that $C_i⪰C_i'→C_{i+1}$ for $0≤i<n$.
Then $C_{end}$ is coverable from $C_{start}$ if and only if there is a lossy run from $C_{start}$ to some $C_n⪰C_{end}$.

A configuration of a $k$-NCS is essentially a tree of depth $k+1$ and can be encoded into a $[k]$-attributed data word as a frame of positions, similar as done to prove Proposition~\ref{prp:linearisation}.
We use an alphabet $Σ$ where every letter $a∈Σ$ encodes, among other information, a $(k+1)$-tuple of states from $Q$, i.e., a possible branch in the tree.
Then a sequence of such letters represents a set of branches that form a tree.
The data valuations represent the information which of the branches share a common prefix.
Further, this representation is interlaced: it only uses odd positions.
The even position in between are used to represent an exact copy of the structure but with distinct data values.
We use appropriate LTL$^↓_{[k]}$ formulae to express this shape.
Figure~\ref{fig:ncs-conf-encoding} shows an example.

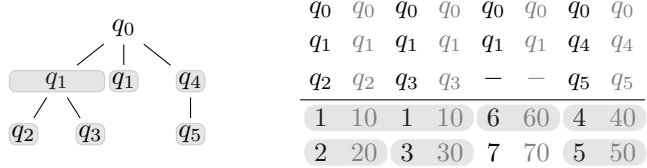
\begin{figure}[t]
  \begin{center}
  \input{fig_ncs-conf-encoding}
  \end{center}
  \caption{Encoding of a $2$-NCS configuration (l.) as $[2]$-attributed data word (r.). 
           Instead of letters from $Σ$ the encoded tuples of states from $Q$ 
           are displayed at every position. %
           \label{fig:ncs-conf-encoding}}
\end{figure}

To be able to formulate the effect of transition rules without using past-time operators we encode lossy runs \emph{reversed}. 
Given that a data word encodes a sequence $C_0C_1…C_n$ of configurations as above we model the (reversed) control flow of the NCS $𝓝=(Q,δ)$ by requiring that every configuration but for the last be annotated by some transition rule $t_j∈δ$ for $0≤j<n$.
The labelling is encoded into the letters from $Σ$ and we impose that this transition sequence actually represents the reversal of a lossy run.
That is, for every configuration $C_j$ in the sequence (for $0≤j<n$) with annotated transition rule $t_j$ there is a configuration $C_{j+1}'$ (not necessarily in the sequence) such that $C_j\stackrel{t_j}{←}C_{j+1}'⪯C_{j+1}$.

For the transition $t_j$ to be executed correctly (up to lossiness) we impose that every branch in $C_j$ must have a corresponding branch in $C_{j+1}$.
Yet, there may be branches in $C_{j+1}$ that have no counterpart in $C_j$ and were thus lost upon executing $t_j$.
Shared data values are now used to establish a link between corresponding branches: 
for every even position in the frame that encodes $C_j$ there must be an odd position in the consecutive frame (thus encoding $C_{j+1}$) with the same data valuation.
To ensure that links are unambiguous we require that every data valuation occurs at most twice in the whole word.
Depending on the effect of the current transition the letters of linked positions are related accordingly.
E.g., for branches not affected at all by $t_j$ the letters are enforced to be equal.
This creates a chain of branches along the run that are identified: an odd position links forward to an even one, the consecutive odd position mimics it and links again forward.

Based on these ideas we can construct a formula satisfied precisely by words encoding a lossy run between particular configurations.
The size of the formula is polynomial in the size of the NCS $𝓝$ and can be built by instantiating a set of patterns while iterating over the transitions and states of $𝓝$, requiring logarithmic space to control the iterations.

\begin{restatable}{theorem}{thmncstofreeze}
\label{thm:ncstofreeze}
  The coverability problem of $k$-NCS can be reduced in logarithmic space to 
  LTL$^↓_{[k]}$ satisfiability.
\end{restatable}

%% file: fig_threads1.tex
\begin{tikzpicture}[
    level 1/.style={level distance =2.5em},
    level 2/.style={level distance =2.5em},
    >=stealth
  ]
  
  \path node [label={[label distance=-1ex] east:$\cmark$}](a){$｛φ_1｝$}
      child [sibling distance =5em] { node {$∅$}
        child {node {$∅$}}
      }
      child [sibling distance =5em]{ 
        node [label={[label distance=-1ex] east:$\cmark$}] {$｛φ_2,ψ_2｝$}
        child [sibling distance =3.5em] {node {$｛ψ_3｝$}}
        child [sibling distance =3.5em] {
          node [label={[label distance=-1ex] east:$\cmark$}] {$｛φ_3｝$}
        }
      } 
    -- ++(3,0) node (b){$∅$}
      child {node {$∅$} 
        child {node {$｛a｝$}}
      }
    -- ++(1.25,-2.5em) node {$→^*$}
    -- ++(2.25,2.5em) node [label={[label distance=-1ex] east:$\cmark$}] (c) {$｛φ_1｝$}
      child [sibling distance =4em] { node {$｛¬↑^2｝$}
        child {node {$｛↑^1｝$}}
      }
      child [sibling distance =7em]{ 
        node [label={[label distance=-1ex] east:$\cmark$}] 
             {$｛φ_2∧ψ_2｝$}
        child [sibling distance =4em] {node {$｛ψ_3｝$}}
        child [sibling distance =4em] {
          node [label={[label distance=-1ex] east:$\cmark$}] {$｛φ_3, ↑^3｝$}
        }
      } 
    -- ++(3.5,0) node (d) {$｛↓^2φ_2｝$}
      child {node {$｛¬b｝$} 
        child {node {$｛a｝$}}
      }
    ;
    
    \path (a-1-1.south west) --node[yshift = -1em] {$C$} (b-1-1.south east);
    \path (c-1-1.south west) --node[yshift = -1em] {$C'$} (d-1-1.south east);
    
    \path (a.north west) node [xshift=1.5ex,anchor = north east] {\color{gray}$^{^1}$};
    \path (a-1.north west) node [xshift=1.5ex,anchor = north east] {\color{gray}$^{^2}$};
    \path (a-1-1.north west) node [xshift=1.5ex,anchor = north east] {\color{gray}$^{^3}$};
    \path (a-2.north west) node [xshift=1.5ex,anchor = north east] {\color{gray}$^{^4}$};
    \path (a-2-1.north west) node [xshift=1.5ex,anchor = north east] {\color{gray}$^{^5}$};
    \path (a-2-2.north west) node [xshift=1.5ex,anchor = north east] {\color{gray}$^{^6}$};
    \path (b.north west) node [xshift=1.5ex,anchor = north east] {\color{gray}$^{^7}$};
    \path (b-1.north west) node [xshift=1.5ex,anchor = north east] {\color{gray}$^{^8}$};
    \path (b-1-1.north west) node [xshift=1.5ex,anchor = north east] {\color{gray}$^{^9}$};

\end{tikzpicture}

%% file: fig_ncs-conf-encoding.tex
\begin{tikzpicture}[
    level 1/.style={level distance =2em},
    level 2/.style={level distance =2em},
    level 3/.style={level distance =2em},
    sibling distance=2.5em
]


  \tikzset{singlenode/.style = {inner sep=0.2ex, outer sep = 0.5ex, rounded corners=2.5pt, fill=gray!20, draw=black!20}}

  \path node[inner sep=0.2ex, outer sep = 0.5ex] (root) {$q_0$} 
    child {node[inner sep=0.2ex, outer sep = 0.5ex] {$q_1$}
            child {node[singlenode] {$q_2$}}
            child {node[singlenode] {$q_3$}}
          }    
    child {node[singlenode] {$q_1$}}
    child {node[singlenode] {$q_4$}
            child {node[singlenode] {$q_5$}}
          };

\path[draw, fill=gray, opacity=0.2, rounded corners=2.5pt] 
  ([yshift=-0.5ex,xshift=0.5ex]root-1.north-|root-1-1.west) rectangle ([yshift=0.5ex,xshift=-0.5ex]root-1.south-|root-1-2.east);


\tikzset{gray columns/.style={column #1/.style={nodes={text=gray}}}}

\matrix[matrix of math nodes, 
        anchor=west, xshift=3em, 
        gray columns/.list={2,4,6,8}
       ] (word) at (root-3.east)
{
q_0 & q_0 & q_0 & q_0 & q_0 & q_0 & q_0 & q_0\\
q_1 & q_1 & q_1 & q_1 & q_1 & q_1 & q_4 & q_4\\
q_2 & q_2 & q_3 & q_3 & -   & -   & q_5 & q_5\\
\hline
1 & 10 & 1 & 10 & 6 & 60 & 4 &40 \\
2 & 20 & 3 & 30 & 7 & 70 & 5 & 50 \\
};

\path[draw=gray, fill=gray, opacity=0.2, rounded corners] 
  ([yshift=-2pt] word-4-1.north west) rectangle ([yshift=2pt] word-4-4.south east)
  ([yshift=-2pt]word-4-5.north west) rectangle ([yshift=2pt]word-4-6.south east)
  ([yshift=-2pt]word-4-7.north west) rectangle ([yshift=2pt]word-4-8.south east)
  ([yshift=-2pt]word-5-1.north west) rectangle ([yshift=2pt]word-5-2.south east)
  ([yshift=-2pt]word-5-3.north west) rectangle ([yshift=2pt]word-5-4.south east)
  ([yshift=-2pt]word-5-7.north west) rectangle ([yshift=2pt]word-5-8.south east);

\end{tikzpicture}

%% file: sec_conclusion.tex
\section{Conclusion}
\label{sec:conclusion}
By Theorem~\ref{thm:freezetoncs} together with Proposition~\ref{prp:ncstopcs} and Theorem~\ref{thm:ncstofreeze} with Theorem~\ref{thm:ncshardness} we can now characterise the complexity of LTL$^↓_{tqo}$ fragments as follows.
\begin{proposition}
Satisfiability of LTL$^↓_A$ over a tree-quasi-ordered attribute set of depth $k$ is in $𝐅_{Ω_{2(k+1)}}$ and $𝐅_{Ω_k}$-hard.
\end{proposition}
Together with Theorem~\ref{thm:undecidability} this completes the proof 
of Theorem~\ref{thm:freeze-decidability-complexity} stating that LTL$^↓_{tqo}$ is the maximal decidable fragment of LTL$^↓_{qo}$ and $𝐅_{ε_0}$-complete.
The result also shows that the complexity of the logic continues to increase strictly with the depth of the attribute ordering.

The logics ND-LTL$^{±}$ were shown to be decidable by reduction to NCS \cite{DBLP:conf/concur/DeckerHLT14}.
Our results thus provide a first upper bound for their satisfiability problem.
Moreover, we derive significantly improved lower bounds by applying the construction to prove Theorem~\ref{thm:ncstofreeze} analogously to ND-LTL$^+$ and, with reversed encoding, to the past fragment ND-LTL$^-$.
A subtle difference is that an additional attribute level is needed in order to express the global data-aware navigation needed to enforce the links between encoded configurations.
\begin{corollary}
  Satisfiability of ND-LTL$^{±}$ with $k+1$ levels is in $𝐅_{Ω_{2(k+1)}}$ and $𝐅_{Ω_k}$-hard. 
\end{corollary}

PCS were proposed as a “master problem” for $𝐅_{ε_0}$ \cite{DBLP:journals/corr/abs-1301-5500} and indeed our upper complexity bounds for NCS rely on them.
However, they are not well suited to prove our hardness results.
This is due to PCS being based on sequences and the embedding ordering while NCS are only based on multisets and the subset ordering.
In a sense, PCS generalise the concept of channels to multiple levels of nesting, whereas NCS generalise the concept of counters.
Hence, we believe NCS are a valuable addition to the list of $𝐅_{ε_0}$-complete models.
They may serve well to prove lower bounds for formalisms that are like Freeze LTL more closely related to the concept of counting.

%% file: apx_freeze_undecidability.tex
\section{Undecidability of LTL$^↓_{qo}$}
\label{sec:apx-freeze-undecidability}

In this section we provide the technical details for establishing undecidability of LTL$^↓_{qo}$. 
Recall Theorem~\ref{thm:undecidability}.

\thmundecidability*

Semi-decidability is obvious when realising that the particular data values in a data word are irrelevant. It suffices to enumerate representatives of the equivalence classes modulo permutations on $Δ$. 

We proceed by first establishing undecidability for a base case with three attributes and generalise it then to an arbitrary number of attributes.

\subsection{Base Case}

\begin{lemma}\label{lem:qo-undecidable}
  For the quasi-ordered set $(A,⊑)$ of attributes with $A=｛\mathsf{x},\mathsf{y},\mathsf{z}｝$ where $\mathsf{x}⊑\mathsf{z}⊒\mathsf{y}$ and $\mathsf{x},\mathsf{y}$ are incomparable the satisfiability problem of LTL$^↓_A$ is undecidable.
\end{lemma}

We can reduce the undecidable \emph{Post's correspondence problem (PCP)} (see, e.g., \cite{DBLP:books/daglib/0000197}) to satisfiability of LTL$^↓_A$. 
We consider an encoding of the problem used in \cite{DBLP:journals/tocl/BojanczykDMSS11}. 
An instance of PCP is given by a finite set $T⊆Σ^*×Σ^*$ of tiles of the form $t=(u,v)$, $u,v∈Σ^*$, over some finite alphabet $Σ$.
  The problem is to decide whether there exists a finite sequence $t_1t_2…t_n = (u_1,v_1)(u_2,v_2)…(u_n,v_n)$ such that the “$u$-part“ and the “$v$-part” coincide, i.e. $u_1u_2…u_n = v_1v_2…v_n$.

The idea in \cite{DBLP:journals/tocl/BojanczykDMSS11} is to encode a sequence of tiles in a word over the alphabet $Σ\dot{∪}\ol{Σ}$, where a distinct copy $\ol{Σ} := ｛\ol{a}｜a∈Σ｝$ is used to encode the $v$-part and letters from $Σ$ encode the $u$-part.
For $v=a_1a_2…∈Σ^*$, $a_i∈Σ$, we let $\ol{v}=\ol{a}_1\ol{a}_2…$.
A sequence of tiles $(u_1,v_1)(u_2,v_2)…$ is then encoded as $\ol{v_1}u_1\ol{v_2}u_2…$.
The switched order of encoding a tile $(u_i,v_i)$ avoids some edge-cases later.

In our setting we require letters to encode additional information and therefor use an alphabet $Γ=(Σ\dot{∪}\ol{Σ})×2^{AP}$ where $AP$ is a finite set of atomic propositions.

To show undecidability of the logic it now suffices to construct from an arbitrary instance of PCP $T$ a formula that expresses sufficient and necessary conditions for a data word $(a_1,𝐝_1)(a_2,𝐝_2)…∈(Γ×Δ^A)^+$ to encode a solution to the PCP instance in terms of the projection $a_1'a_2'…∈(Σ\dot{∪}\ol{Σ})^+$ of $w$ where $a_i=(a_i',m_i)$ for some $m_i∈2^{AP}$.

In order to translate the conditions given in \cite{DBLP:journals/tocl/BojanczykDMSS11} in terms of first-order logic formulae, past-time operators would be necessary.
To avoid these we use the following two ideas from \cite{DBLP:conf/lics/DemriFP13}.
Let, for a sequence of tiles $(u_1,v_1)(u_2,v_2)…(u_n,v_n)∈T^*$ denote $u:=u_1u_2…u_n$ and $v:=v_1v_2…v_n$.
First, we assume $AP$ to contain two propositions $e$ and $o$ that shall mark even and odd positions, respectively, in $u$ and $v$.
Second, we use a variant of PCP that imposes additional restrictions on a valid solution $t_1t_2…t_n∈T^*$: 
\begin{itemize}
  \item The initial tile is fixed to $t_1=\hat{t}=(\hat{u},\hat{v})∈T$ with $|\hat{u}|>1$ and $|\hat{v}|>2$,
  \item for every strict prefix $t_1t_2…t_i$, $i<n$, the $u$-part must be strictly \emph{shorter} than the $v$ part and
  \item $|u|$ (i.e., the length of the solution) is odd.
\end{itemize}

The first condition turns the problem into what is called a \emph{modified} PCP in \cite{DBLP:books/daglib/0000197} and shown undecidable there by a reduction from the halting problem of Turing machines. 
It was observed in \cite{DBLP:conf/lics/DemriFP13} that this encoding of Turing machines actually guarantees that the length of the $u$-part is always shorter.
As pointed out in \cite{DBLP:journals/tocl/BojanczykDMSS11}, the last condition is not an actual restriction because adding a tile $(\$x,\$y)$ for every tile $(x,y)∈T$ yields a PCP instance that has an odd solution if and only if the original instance had any solution.

We can now adjust the set of conditions from \cite{DBLP:journals/tocl/BojanczykDMSS11} such that they can be formulated in LTL$^↓_A$ and impose the additional restrictions on a solution.

For easier reading, we use letters $a∈(Σ\dot{∪}\ol{Σ})$ in formulae to denote $⋁_{p∈AP}(a,p)$ and 
propositions $p∈AP$ to denote $⋁_{(a,m)∈Γ｜p∈m}(a,m)$.
Also, we use $Σ$ and $\ol{Σ}$ to denote the formulae $⋁_{a∈Σ}a$ and $⋁_{\ol{a}∈\ol{Σ}}\ol{a}$, respectively. We write $\Fφ$ for $\true\Uφ$.

\subparagraph{Global structure.}
Let $AP$ contain a proposition $\textsf{end}$ that we use to mark the end of a sequence of letters that encode a tile. 
For a tile $t=(b_1…b_m,a_1…a_n)∈T$ where $b_i,a_i∈Σ$ let 
\[
  φ_t := \left(⋀_{i=1}^n \X^{i-1} \ol{a_i}\right)
         ∧ \left(⋀_{i=1}^{m} \X^{n+i-1} b_i \right)
         ∧ \left(⋀_{i=1}^{n+m-1} \X^{i-1} ¬\textsf{end}\right) 
         ∧ \X^{n+m-1} \textsf{end}
\]
be the formula expressing that the following positions encode the tile $t$.

The global structure of a word that correctly encodes a sequence of tiles as described above can now be expressed in terms of the following conditions.
  \begin{itemize}
  \item The word encodes a sequence of tiles from $T$, starting with $\hat{t}$:
        \begin{equation}\label{eqn:tile-chain}
          φ_{\hat{t}} ∧ \G(\textsf{end} → \WX ⋁_{t∈T} φ_t)
        \end{equation}
  \item The even and odd positions on the substrings $u$ and $\ol{v}$ of a solution are marked correctly with $e$ and $o$, respectively.
        For $\hat{t}=(u_1,v_1)$ the formula
    \begin{equation}\label{eqn:evenodd-base}
      \G(e ↔ ¬o) ∧ o ∧ \X^{|v_1|} o,
    \end{equation}
    expresses the exclusiveness of markings $e$ and $o$ and specifies the first $v$-position and the first $u$-position in the encoding to be odd.
    Then, the alternation of these markings in the subsequence encoding $u$ is expressed by 
    \begin{equation}\label{eqn:evenodd-induction-u}
    \begin{split}
        \G( (Σ ∧ o) → \X(¬Σ\U(Σ∧e) ∨ \G(¬Σ)) )\\
      ∧ \G( (Σ ∧ e) → \X(¬Σ\U(Σ∧o) ∨ \G(¬Σ)) )
    \end{split}
    \end{equation}
    and the alternation of markings in the subsequence encoding $v$ by
    \begin{equation}\label{eqn:evenodd-induction-v}
    \begin{split} 
        \G((\ol{Σ} ∧ o) → \X(¬\ol{Σ}\U(\ol{Σ}∧e) ∨ \G(¬\ol{Σ})) )\\
      ∧ \G((\ol{Σ} ∧ e) → \X(¬\ol{Σ}\U(\ol{Σ}∧o) ∨ \G(¬\ol{Σ})) )
    \end{split}
    \end{equation}
  \end{itemize}

Let $Φ_T$ be the conjunction of Formulae~\ref{eqn:tile-chain},\ref{eqn:evenodd-base}, \ref{eqn:evenodd-induction-u} and \ref{eqn:evenodd-induction-v}.

\subparagraph{Chaining $u$ and $v$.}

In order to connect the positions belonging to the subword $u$ and the subword $v$ we link consecutive position by a shared data value as depicted in Figure~\ref{fig:pcp-encoding}.

\begin{figure}
  \input{fig_pcp-encoding}
  \caption{Structure of a data word encoding a sequence of tiles $t_1t_2t_3$ 
           with $t_1=(ab,abc)$, $t_2=(cc,cab)$, $t_3=(a,ε)$.
           \label{fig:pcp-encoding}}
\end{figure}

For the subword encoding $u$ the structure is imposed by the following constraints.
\begin{itemize}
  \item Each data value occurs at most twice and not in both attributes $\mathsf{x}$ and $\mathsf{y}$:
        \begin{equation}\label{eqn:values-twice} \begin{split}
          \G(Σ → ↓^\mathsf{x} ((¬\F↑^\mathsf{y}) ∧ ¬\X\F(Σ ∧ ↑^\mathsf{x} ∧ \X\F(Σ ∧ ↑^\mathsf{x})))\\
          ∧ \G(Σ → ↓^\mathsf{y} ((¬\F↑^\mathsf{x}) ∧ ¬\X\F(Σ ∧ ↑^\mathsf{y} ∧ \X\F(Σ ∧ ↑^\mathsf{y})))
        \end{split}\end{equation}
  \item At any odd position (except for the last) the data value for attribute 
        $\mathsf{x}$ occurs again in $\mathsf{x}$ at an even future position and the value of
        attribute $\mathsf{y}$ does never occur again.
        At any even position, the same holds vice versa for $\mathsf{y}$ and $\mathsf{x}$.
        \begin{equation}\label{eqn:repeating-values-u}
          \G\left(
            \begin{array}{ccl}
                & (Σ ∧ o ∧ \X\FΣ) &→ (↓^\mathsf{x}\F(↑^\mathsf{x} ∧ Σ ∧ e)) ∧ ¬↓^\mathsf{y}\X\F(↑^\mathsf{y} ∧\ Σ)\\
              ∧ & (Σ ∧ e)         &→ (↓^\mathsf{y}\F(↑^\mathsf{y} ∧ Σ ∧ o)) ∧ ¬↓^\mathsf{x}\X\F(↑^\mathsf{x} ∧\ Σ)
            \end{array}\right)
        \end{equation}
\end{itemize}
The same restrictions can be formulated analogously for the subword that encodes $v$.
\begin{equation}\label{eqn:repeating-values-v}\begin{aligned}
  & \G(\ol{Σ} → ↓^\mathsf{x} ((¬\F↑^\mathsf{y}) ∧ ¬\X\F(\ol{Σ} ∧ ↑^\mathsf{x} ∧ \X\F(\ol{Σ} ∧ ↑^\mathsf{x})))\\
  ∧ &\G(\ol{Σ} → ↓^\mathsf{y} ((¬\F↑^\mathsf{x}) ∧ ¬\X\F(\ol{Σ} ∧ ↑^\mathsf{y} ∧ \X\F(\ol{Σ} ∧ ↑^\mathsf{y})))\\
  ∧ &\G\left(
      \begin{array}{ccl}
         & (\ol{Σ} ∧ o ∧ \X\F\ol{Σ}) &→ (↓^\mathsf{x}\F(↑^\mathsf{x} ∧ \ol{Σ} ∧ e)) ∧ ¬↓^\mathsf{y}\X\F(↑^\mathsf{y} ∧\ \ol{Σ})\\
       ∧ & (\ol{Σ} ∧ e)         &→ (↓^\mathsf{y}\F(↑^\mathsf{y} ∧ \ol{Σ} ∧ o)) ∧ ¬↓^\mathsf{x}\X\F(↑^\mathsf{x} ∧\ \ol{Σ})
      \end{array}\right)
\end{aligned}\end{equation}

Let $Φ_{\mathsf{chain}}$ denote the conjunction of the formulae from Equations~\ref{eqn:values-twice},~\ref{eqn:repeating-values-u} and~\ref{eqn:repeating-values-v}.

To formalise the guarantee on the structure that we obtain from these constraints let 
\[
  w=v_1…v_{m_1}u_1…u_{n_1}v_{m_1+1}…v_{m_1+m_2}u_{n_1+1}…u_{n_1+n_2}…
\]
for $v_i,u_i∈Γ×Δ^A$ be the encoding of a sequence of tiles $t_1t_2…∈T^+$ with $w⊧Φ_T ∧ Φ_{\mathsf{chain}}$.
Further let $u=u_1u_2…=(a_1,𝐝_1)(a_2,𝐝_2)…$ and $v=v_1v_2…=(b_1,𝐞_1)(b_2,𝐞_2)…$ be the subwords of $w$ encoding the $u$-part and the $v$-part of the tile sequence, respectively.
\begin{lemma}\label{lem:chaining}
  Let $i<k$ be a position in the subword $u$ of $w$ with length $|u|=k$. 
  \begin{enumerate}
    \item If $i$ is \emph{odd} then $𝐝_i(\mathsf{x})=𝐝_{i+1}(\mathsf{x})$.
    \item If $i$ is \emph{even} then $𝐝_i(\mathsf{y})=𝐝_{i+1}(\mathsf{y})$.
    \item For all positions $i,j$ on $u$ we have 
          $(𝐝_i(\mathsf{x}) =𝐝_j(\mathsf{x}) ∧  𝐝_i(\mathsf{y}) = 𝐝_j(\mathsf{y})) ⇒ i = j$.
  \end{enumerate}
  The same holds analogously for $v$.
\end{lemma}

\begin{proof}
  Given the formula $Φ_T$ it is easy to see that the even and odd positions in the subwords $u$ and $v$ are correctly marked by the respective propositions.
We present the proof only for $u$ since it is identical for $v$.
  
\textbf{1.+2.)}
  Assume $u$ has length $k$. We proceed by induction on the positions $i$, backward from $k-1$ down to 1. 

  \emph{Base case $i=k-1$.}\quad
  The length $k$ of $u$ needs to be odd, otherwise there is no even future position and Formula~\ref{eqn:repeating-values-u} is violated.
  Hence $i=k-1$ is even and Formula~\ref{eqn:repeating-values-u} ensures that the value $𝐝_{k-1}(\mathsf{y})$ is repeated in attribute $\mathsf{y}$, leaving $𝐝_k$ as only choice.
  
  \emph{Induction.}\quad
  Assume for $i+1<k-1$ the statement holds. Assume $i$ is odd. Since the position $i+1$ exists, by Formula~\ref{eqn:repeating-values-u} there is a position $j>i$ such that $𝐝_i(\mathsf{x})=𝐝_j(\mathsf{x})$.
  Now for every even $j>i+2$ the induction hypothesis holds for $j-1$, being odd. I.e. $𝐝_{j-1}(\mathsf{x})=𝐝_j(\mathsf{x})$. 
  Since the value $𝐝_j(\mathsf{x})$ can only occur at most twice for $\mathsf{x}$ in $u$ (Formula~\ref{eqn:values-twice}), we have $𝐝_i(\mathsf{x})≠𝐝_j(\mathsf{x})$, leaving $j=i+1$ as only choice.

  Assume $i$ is even. 
  Formula~\ref{eqn:repeating-values-u} requires that $𝐝_i(\mathsf{y})=𝐝_j(\mathsf{y})$ for some odd $j>i$.
  Again, for any odd $j>i+1$ the induction hypothesis holds for $j-1≥i+1$, being even. 
  We have $𝐝_{j-1}(\mathsf{y}) =𝐝_j(\mathsf{y})$ and thus $𝐝_i(\mathsf{y})≠𝐝_j(\mathsf{y})$ for any odd $j>i+1$.
  Therefore only position $j=i+1$ remains to carry the same value for $\mathsf{y}$.
  
\textbf{3.)} Let $𝐝_i(\mathsf{x})=𝐝_j(\mathsf{x})$ and $𝐝_i(\mathsf{y})=𝐝_j(\mathsf{y})$. Assume $i$ is odd.
Then, $𝐝_i(\mathsf{x})=𝐝_{i+1}(\mathsf{x})$ (see above) and, by Formula~\ref{eqn:repeating-values-u}, $∀_{i'>i+1}: 𝐝_{i'}(\mathsf{x})≠𝐝_i(\mathsf{x})$. 
Hence $i≤j≤i+1$.
Moreover, $∀_{i'>i}: 𝐝_{i'}(\mathsf{y})≠𝐝_i(\mathsf{y})$ and thus $j=i$.
For even $i$ the argument holds analogously.  
\end{proof}

\subparagraph{Synchronising $u$ and $v$.}

Now that the encoding of $u$ and $v$ is set up, we enforce that 
\begin{enumerate}
  \item every position in $v$ matches a unique position in $u$,
  \item the first $v$ position matches the first $u$ position and
  \item for any two consecutive positions in $v$ the corresponding
        matching positions in $u$ are also consecutive. Finally,
  \item the last position in $v$ matches the last position in $u$.
\end{enumerate}

This is accomplished by the formula $Φ_{\textsf{sync}}$ being the conjunction of the three formulae
  \begin{gather}
      (↓^\mathsf{z}\X^{|\hat{v}|}↑^\mathsf{z}),\label{eqn:init-positions-match}\\
      ⋀_{\ol{a}∈\ol{Σ}} \G(\ol{a} → ↓^\mathsf{z}(\X\F(a\ ∧ ↑^\mathsf{z}))),\label{eqn:z-match}\\
      \G\big((\ol{Σ} ∧ ¬\X\F\ol{Σ}) → ↓^z(\X\F(↑^z ∧ ¬\X\true))\big),\label{eqn:last-pos-match}
  \end{gather}
where $\hat{v}$ is part of the fixed initial tile $\hat{t}=(\hat{u},\hat{v})$.
The formula specifies that 
\begin{itemize}
  \item each set $｛𝐞_i(\mathsf{x}),𝐞_i(\mathsf{y})｝$ of values occurring at some position $i$ in $v$ occurs again at a position in $u$ with the same (encoded),
  \item the data values at the first positions in $v$ and $u$ coincide and
  \item the data values at the last positions in $v$ and $u$ coincide.
\end{itemize}

For easier reading let $\ol{(a,𝐝)} := (\ol{a},𝐝)∈\ol{Σ}×Δ$ for $(a,𝐝)∈Σ×Δ$.
The essential observation is now the following.

\begin{lemma}\label{lem:sync}
Let $w⊧Φ_T ∧ Φ_{\mathsf{chain}} ∧ Φ_{\mathsf{sync}}$ be a data word and $w'$ its projection to the alphabet $(Σ∪\ol{Σ})×Δ$.
Let $u=u_1…u_{|u|}$ and $v=v_1…v_{|v|}$ be the maximal subwords of $w'$ over $Σ×Δ$ and $\ol{Σ}×Δ$, respectively.
Then, for all $0<i≤|v|$ we have $i≤|u|$ and $v_i=\ol{u_i}$.
\end{lemma}
That is, the $i$-th position in the $v$-part of $w$ corresponds to the $i$-th position in the $u$-part and thus $v$ encodes a prefix of the $u$-part.
Since the last position in $v$ must corresponds to the last position in $u$ (Equation~\ref{eqn:last-pos-match}), $v$ and $u$ must encode the \emph{same} sequence of letters and $w$ therefore a solution to the PCP $T$.
On the other hand, given a solution for $T$, we can easily encode it according to the scheme depicted in Figure~\ref{fig:pcp-encoding} where corresponding positions in $u$ and $v$ can be linked appropriately. 
This encoding satisfies (by construction) all the constraints imposed by the formulae above.
Hence, by proving Lemma~\ref{lem:sync} we complete the proof of Lemma~\ref{lem:qo-undecidable}.

\begin{proof}[Lemma~\ref{lem:sync}]
Let $u=(a_1,𝐝_1)…(a_ℓ,𝐝_ℓ)$ and $v=(b_1,𝐞_1)…(b_k,𝐞_k)$.
We proceed by induction on $i$.

\emph{Base case ($i=1$).} For the initial tile $\hat{t}=(\hat{u},\hat{v})$ we assumed that $|\hat{v}|≥1$ and $|\hat{u}|≥1$ so $1$ is a position in $u$ as well as in $v$. 
Equation~\ref{eqn:init-positions-match} requires that $v_1=\ol{u_1}$.

\emph{Induction ($i>1$).} Assume $i≤|v|$ is a position in $v$.
Thus, all $0<j<i$ are positions in $v$ and by the induction hypothesis (IH) also positions in $u$ with $v_j=u_j$.

Assume $i$ is odd. We have
\[
  𝐝_i(\mathsf{y}) 
  \stackrel{\text{Lem.~\ref{lem:chaining}}}{=} 𝐝_{i-1}(\mathsf{y}) 
  \stackrel{\text{IH}}{=} 𝐞_{i-1}(\mathsf{y})
  \stackrel{\text{Lem.~\ref{lem:chaining}}}{=} 𝐞_i(\mathsf{y})
\]

Let $v_i=(b_i,𝐞_i) = (\ol{a},𝐞_i)$ for some $a∈Σ$. 
By Equation~\ref{eqn:z-match} thre must be a position $u_j=(a_j,𝐝_j)$ in $u$ where $a_j = a$ and $𝐝_j=𝐞_i$.

By $Φ_{\mathsf{chain}}$, there can be at most two positions $j$ with $𝐝_j(\mathsf{y})=𝐞_i(\mathsf{y})$ and we have already $𝐝_{i-1}(\mathsf{y}) = 𝐝_i(\mathsf{y}) = 𝐞_i(\mathsf{y})$.
Hence, $j∈｛i,i-1｝$.
That is, either $\ol{u_i}=v_i$ or $\ol{u_{i-1}}=v_i$.
The latter can be excluded since 
\[
  \ol{u_{i-1}} \stackrel{\text{IH}}{\ =\ } v_{i-1} \stackrel{\text{Lem.~\ref{lem:chaining}}}{=} v_i.
\]
If $i$ is assumed to be even the same arguments apply when exchanging attribute $\mathsf{y}$ by attribute $\mathsf{x}$.
\end{proof}

\begin{remark}
      Notice that we do not rely on using an until operator.
      Instead, we can replace it by a bounded version $\U^{≤k}$ that in turn can be replaced by a finite unfolding only using nested \X operators.
      The relevant range can be bound by the length of the tiles in $T$ as
      \[
        k≥2·\max｛|r|｜∃_s(r,s)∈T ∨ (s,r)∈T｝.
      \]
      Thus, we take $k$ to be at least as large as the longest consecutive pair $\ol{v_i}u_i$ or $u_i\ol{v_{i+1}}$ in the encoding of a solution could possibly be.
      This is a bound on the distance between two position in the encoding that are consecutive in $u$ or $v$.
      Hence, only the operators \X and \F are essential.
\end{remark}

\subsection{General Case}

We can now complete the proof of Theorem \ref{thm:undecidability}.

\begin{proof}[Theorem \ref{thm:undecidability}]
Lemma~\ref{lem:qo-undecidable} established undecidability for the essential case of a non-tree-quasi-ordering.
It remains to conclude that this results generalises to arbitrary non-tree-quasi-orderings.

Let $(A,⊑)$ be the quasi-ordering defined in Lemma~\ref{lem:qo-undecidable}.
First of all, $(A,⊑)$ is not a tree-quasi-ordering since the downward-closure $\cl(\mathsf{z})$ of $\mathsf{z}$ is not quasi-linear (total).
Moreover, every non-tree-quasi-ordering $(A',⊑')$ has a subset that is isomorphic to $A$:
By definition $A'$ must contain an element $z'$ of which the downward-closure is not quasi-linear and must hence contain two incomparable elements $x'⊑z'$ and $y'⊑z'$. 
Hence from now on we assume w.l.o.g.\ that $A⊆A'$ by identifying $\mathsf{x},\mathsf{y},\mathsf{z}$ with $x',y',z'$, respectively.

We now show that the formula $Φ$ constructed to prove Lemma~\ref{lem:qo-undecidable} is satisfiable over $A$-attributed data words if and only if it is satisfiable when being interpreted over $A'$-attributed data words.
  
\subparagraph{($⇒$)} Consider an $A$-attributed data word $w$ satisfying $Φ$.
Choose a data value $e∈Δ$ that does not occur in $w$ and extend $w$ to an $A'$-attributed data word $w'$ by assigning $e$ to every attribute $p∈A'∖A$ at every position in $w'$.
  This does not change the satisfaction relation because $Φ$ still only uses attributes from $A$ and the evaluation of formulae $↑^r$ for $r∈A$ is not affected:
For $w=(a_1,𝐝_1)…(a_n,𝐝_n)$, $w'=(a_1,𝐝'_1)…(a_n,𝐝'_n)$, $0<i≤j≤n$, $r'∈A$ we have  
\[
  (w,j,𝐝_i|_r) ⊧↑^{r'} \quad ⇔ \quad (w',j,𝐝'_i|_r) ⊧↑^{r'}.
\]
\begin{enumerate}[i)]
  \item Let $r=r'∈｛\mathsf{x},\mathsf{y}｝$. 
    Notice that $(w,j,𝐝_i|_r) ⊧↑^{r}$ iff $𝐝_i(r) = 𝐝_j(r)$.
    Then $𝐝_i(r) = 𝐝_j(r)$ implies that 
    $∃_{p∈A'}: 𝐝'_i|_r|_p≃ 𝐝'_j|_r$ since for $p=r$ the
    restrictions are isomorphic as all other attributes in $\cl(r)$ are always 
    mapped to $e$.
    Conversely, if $∃_{p∈A'}: 𝐝'_i|_r|_p≃ 𝐝'_j|_r$ then it 
    can only be the case for some $p$ such that $\cl(p) = \cl(r)$.
    Since $𝐝_j(q)=𝐝_i(q)=e≠𝐝_i(r)$ for all $q∈\cl(r)∖｛r｝$ the valuations can only be isomorphic if $𝐝_j(r) = 𝐝_i(r)$.    

  \item Let $r=r'=\mathsf{z}$. We have that  $(w,j,𝐝_i|_\mathsf{z}) ⊧↑^{\mathsf{z}}$ iff 
        $𝐝_i(\mathsf{z}) = 𝐝_j(\mathsf{z})$ and $｛𝐝_i(\mathsf{x}),𝐝_i(\mathsf{y})｝=｛𝐝_j(\mathsf{x}),𝐝_j(\mathsf{y})｝$.
    In our case, the models of $Φ$ only admit disjoint values for attributes 
    $\mathsf{x}$ and $\mathsf{y}$ (cf.\ Formulae~\ref{eqn:values-twice} and~\ref{eqn:repeating-values-v} in the proof of Lemma~\ref{lem:qo-undecidable}).
    Thus, $𝐝_i(\mathsf{x}) = 𝐝_j(\mathsf{x})$ and $𝐝_i(\mathsf{y}) = 𝐝_j(\mathsf{y})$.
    This further implies $∃_{p∈A'}: 𝐝'_i|_\mathsf{z}|_p≃ 𝐝'_j|_\mathsf{z}$ witnessed by choosing $p=\mathsf{z}$ since all other attributes are evaluated to $e$ by $𝐝'_i$ and $𝐝_j'$.
    Moreover, the opposite direction holds for the same reason.
    
  \item Let $r'=\mathsf{x}$ and $r=\mathsf{y}$ or vice versa. Again 
        $(w,j,𝐝_i|_r) ⊧↑^{r'}$ iff $𝐝_i(r) = 𝐝_j(r')$, which however, 
        cannot be true due to $\mathsf{x}$ and $\mathsf{y}$ being assigned disjoint sets of values 
        in every model.
        On the other hand, assume there is $p∈A'$ s.t.\ 
        $𝐝'_i|_r|_p≃ 𝐝'_j|_{r'}$. 
        Clearly the witnessing isomorphism must map $r$ to $r'$ since they are 
        not assigned the value $e$.
        Then, however, $𝐝_i(r) = 𝐝_j(r')$ which violates $Φ$.          
\end{enumerate}
The remaining cases do not occur in $Φ$ (and would evaluate to false anyway).
We conclude that if $w$ is a model for $Φ$ then $w'$ is as well.

\subparagraph{($⇐$)} Consider an $A'$-attributed data word $w'$ satisfying $φ$ and let 
$\ol{Δ_{w'}}⊆Δ$
be an enumerable set of data values not occurring in $w'$.
Let $f: Δ_⊥^{A'}\!\!/\!\!≃\ ↪ \ol{Δ_{w'}}$ be an injection from the $≃$-equivalence classes of data valuations to data values uniquely representing them.
We can then construct an $A$-attributed model $w$ for $φ$ from $w'$ by erasing all attributes except for $\mathsf{x},\mathsf{y},\mathsf{z}$ and let $𝐝_i(p) := f([𝐝_i'|_p]_≃)$ for $p∈A$ where $[𝐝]_≃$ denotes the $≃$-equivalence class of a data valuation $𝐝$.
Intuitively, at any position in $w'$, we just collapse the structure of data values to a single one representing its equivalence class.
By similar arguments as above, we can again show that 
\[
  (w,j,𝐝_i|_r) ⊧↑^{r'} \quad ⇔ \quad (w',j,𝐝'_i|_r) ⊧↑^{r'}.
\]
\begin{enumerate}[i)]
  \item For $r=r'$ we have that $∃_{p∈A'}: 𝐝'_i|_r|_p≃ 𝐝'_j|_r$ iff $𝐝'_i|_r≃ 𝐝'_j|_r'$ iff $[𝐝'_i|_r]_≃ = [𝐝'_j|_r]_≃$ iff 
  $𝐝_i(r)=𝐝_j(r)$.
  \item For $r=\mathsf{x}$ and $r'=\mathsf{y}$ or vice versa $∃_{p∈A'}: 𝐝'_i|_r|_p≃ 𝐝'_j|_r$ cannot be true in a model for $φ$ and this being false implies equally $𝐝_i(r)≠𝐝_j(r')$.  
\end{enumerate}
Again, other cases do not apply.
\end{proof}

%% file: fig_pcp-encoding.tex
\begin{tikzpicture}

\matrix[matrix of math nodes,
        row 6/.style={nodes={text=gray}}
       ] (word) {
	\ &[1em]	v_1	&	v_2	&	v_3	&	u_1	&	u_2	&	v_4	&	v_5	&	v_6	&	u_3	&	u_4	&	u_5	\\
\hline																							
AP	&	\textsf{o}	&	\textsf{e}	&	\textsf{o}	&	\textsf{o}	&	\textsf{e}, \textsf{end}	&	\textsf{e}	&	\textsf{o}	&	\textsf{e}	&	\textsf{o}	&	\textsf{e},\textsf{end}	&	\textsf{o},\textsf{end}	\\
Σ/\ol{Σ}	&	\ol{a}	&	\ol{b}	&	\ol{c}	&	a	&	b	&	\ol{c}	&	\ol{a}	&	\ol{b}	&	c	&	c	&	a	\\
x	&	2	&	2	&	4	&	2	&	2	&	4	&	6	&	6	&	4	&	4	&	6	\\
y	&	1	&	3	&	3	&	1	&	3	&	5	&	5	&	7	&	3	&	5	&	5	\\
z	&	10&	20&	30&	10&	20&	40&	50&	70&	30&	40&	50\\
};

\draw (word-1-2.north-|word-3-1.east) -- (word-6-1.south-|word-3-1.east);


\draw[rounded corners] ([yshift=-2pt]word-4-2.north west) rectangle ([yshift=2pt]word-4-3.south east);
\draw[rounded corners] ([yshift=-2pt]word-5-3.north west) rectangle ([yshift=2pt]word-5-4.south east);
\draw[rounded corners] ([yshift=-2pt]word-4-4.north east) 
      -- ([yshift=-2pt]word-4-4.north west) 
      -- ([yshift=2pt]word-4-4.south west)
      -- ([yshift=2pt]word-4-4.south east);
\draw[rounded corners] ([yshift=-2pt]word-4-7.north west) 
      -- ([yshift=-2pt]word-4-7.north east) 
      -- ([yshift=2pt]word-4-7.south east)
      -- ([yshift=2pt]word-4-7.south west);

\draw[rounded corners] ([yshift=-2pt]word-5-7.north west) rectangle ([yshift=2pt]word-5-8.south east);
\draw[rounded corners] ([yshift=-2pt]word-4-8.north west) rectangle ([yshift=2pt]word-4-9.south east);

\draw[rounded corners] ([yshift=-2pt]word-5-9.north east) 
      -- ([yshift=-2pt]word-5-9.north west) 
      -- ([yshift=2pt]word-5-9.south west)
      -- ([yshift=2pt]word-5-9.south east);

\draw[rounded corners, densely dotted] ([yshift=-2pt]word-4-5.north west) rectangle ([yshift=2pt]word-4-6.south east);
\draw[rounded corners, densely dotted] ([yshift=-2pt]word-5-6.north east) 
      -- ([yshift=-2pt]word-5-6.north west) 
      -- ([yshift=2pt]word-5-6.south west)
      -- ([yshift=2pt]word-5-6.south east);
\draw[rounded corners, densely dotted] ([yshift=-2pt]word-5-10.north west) 
      -- ([yshift=-2pt]word-5-10.north east) 
      -- ([yshift=2pt]word-5-10.south east)
      -- ([yshift=2pt]word-5-10.south west);

\draw[rounded corners, densely dotted] ([yshift=-2pt]word-4-10.north west) rectangle ([yshift=2pt]word-4-11.south east);
\draw[rounded corners, densely dotted] ([yshift=-2pt]word-5-11.north west) rectangle ([yshift=2pt]word-5-12.south east);

\draw[rounded corners, densely dotted] ([yshift=-2pt]word-4-12.north east) 
      -- ([yshift=-2pt]word-4-12.north west) 
      -- ([yshift=2pt]word-4-12.south west)
      -- ([yshift=2pt]word-4-12.south east);
\end{tikzpicture}

%% file: apx_ncs_upperbound.tex
\section{Nested Counter Systems}

\subsection{Upper Bound for NCS Coverability}
\label{apx:ncs-upper-bound}

Recall Proposition~\ref{prp:ncstopcs}.

\prpncstopcs*

The statement can be proven by a direct reduction to coverability (equivalently, control-state reachability) in priority channel systems (PCS) that we briefly recall from \cite{DBLP:journals/corr/abs-1301-5500} in the following.

\subparagraph{Priority Channel Systems.} 

PCS can be defined over so called \emph{generalised priority alphabets}.
Given a \emph{priority level} $d∈ℕ$ and a well-quasi-ordering $(Γ,≤_Γ)$ a generalised priority alphabet is a set $Σ_{d,Γ} := ｛(a,w)｜0≤w≤d, w∈Γ｝$.
Then, a PCS is a tuple $S = (Σ_{d,Γ}, \mathsf{Ch}, Q, Δ)$, where $\mathsf{Ch}$ is a finite set of \emph{channel names}, $Q$ is a finite set of \emph{control states} and $Δ⊆Q×\mathsf{Ch}×｛!,?｝×Σ_{d,Γ}×Q$ is a set of \emph{transition rules}.
The \emph{semantics} of PCS is defined as a transition system over configurations $\mathrm{Conf}_S := Q×(Σ_{d,Γ}^*)^{\mathsf{Ch}}$ consisting of a control state and a function assigning to every channel a sequence of messages (letters from the generalised priority alphabet) it contains.
A PCS can either execute one of its transition rules or an internal “lossy” operation called a \emph{superseding step}.
A (writing) transition rule of the form $(q,c,!,(a,w),q')$ is performed by changing the current control state $q$ to $q'$ and appending the letter $(a,w)$ to the content of channel $c$.
A (reading) transition rule of the form $(q,c,?,(a,w),q')$ is performed by changing the current control state $q$ to $q'$ and removing the letter $(a,w)$ from the first position of channel $c$.
An internal superseding step is performed by overriding a letter by a subsequent letter with higher or equal priority, i.e.\ the channel content $(a_1,w_1)…(a_i,w_i)(a_{i+1},w_{i+1})…(a_k,w_k)$ with $w_i≤w_{i+1}$ can be replaced by $(a_1,w_1)…(a_{i-1},w_{i-1})(a_{i+1},w_{i+1})…(a_k,w_k)$.

\subparagraph{Encoding.} 

The semantics of NCS is defined in terms of rewriting rules on configurations represented as terms.
A PCS can simulate this semantics by keeping the top-level state in its finite control
and storing the term representation of the nested multiset (with an additional marker at the end) in a single channel.
The rewriting rules of the semantics can be applied by alternately reading from and writing to that channel.
To simulate one step of the NCS the PCS guesses the rule that should be applied and then starts to loop one time through the channel.
Going through the channel, the PCS guesses the positions where the rule matches and at the same time writes the corresponding messages back to the channel.
During one iteration, the PCS has to keep track of hte guessed transition and up to which nesting level it already has been applied.

As PCS are lossy the only objective is to ensure that lossiness with respect to the PCS semantics corresponds to descending with respect to $⪯$ for the encoded NCS configurations.
This can be easily ensured by encoding the nesting structure of an NCS configuration using priorities where the highest priority corresponds to the outermost nesting level.
E.g., the 3-NCS configuration 
\[ q_0(
   q_1 
 + q_1(q_2 + q_2) 
 + q_1(q_2 + q_2 + q_3(q_4)))
\]
can be encoded as 
\[
(q_1,2)(q_1,2) (q_2,1) (q_2,1) (q_1,2)(q_2,1) (q_2,1) (q_3,1) (q_4,0) (\$,2)
\]
while $q_0$ is encoded using the control state and $\$$ marks the end of the encoding.
A superseding step then always corresponds to removing an element from an innermost multiset.

\subparagraph{Complexity.}
The above encoding gives us a PCS with $k$ priorities (maximal priority $d = k-1$), one channel, a number of control states $s$ polynomial in $k$ and the number of states and transitions of the NCS and a size of the alphabet $Γ$ linear in the number of the states of the NCS.

The upper bound on the complexity of PCS control-state reachability is proved in \cite{DBLP:journals/corr/abs-1301-5500} by providing a bound on the so-called length function $L_{Σ^*_{p,Γ}}$.
It measures the length controlled bad sequences in the well-quasi-ordered set of channel configurations $Σ^*_{p,Γ}$.
Specifically, the proof of \cite[Corollary~4.2]{DBLP:journals/corr/abs-1301-5500} provides a bound on the length function for the well-quasi-ordering $Σ^*_{k-1,Γ}$ on PCS configurations with a single channel: $L_{Σ^*_{k-1,Γ}}(n) ≤ H^{(Ω_{2k+1})_{|Γ|}}(n)$. 
Taking into account the control states, following \cite[Section~4.3]{DBLP:journals/corr/abs-1301-5500}, this yields a bound $H^{(Ω_{2k+1})_{|Γ|}·s}(n)$ on the length function for configurations of the PCS that we construct from an NCS as outlined above.
The coverability problem of $k$-NCS is hence contained in $𝐅_{(Ω_{2k})_{|Γ|}+1}$.
For simplicity, we use the larger class $𝐅_{Ω_{2k}}$ in the statement of Proposition~\ref{prp:ncstopcs}.

%% file: apx_ncs_lowerbound.tex
\subsection{Lower Bound for NCS Coverability}\label{apx:ncshardness}
In this section we give the detailed constructions proving Theorem~\ref{thm:ncshardness}.
As we have discussed above, we need a construction fulfilling Lemma~\ref{lem:ncshardness1} and Lemma~\ref{lem:ncshardness2}.

\subparagraph{Auxilary operations.}
To this end, we extend the NCS with two auxiliary operations $\cp$ and $\min$.
The semantics of the operation $(q_1, …, q_l) \cp (q_1', …, q_l')$ can be given by the rewriting rule
\begin{align*}
  &q_1(X_1 + q_2(X_2 + … q_l(X_l)) + q'_2(X'_2 + … q'_{l-1}(X'_{l-1})))\\
→ &q'_1(X_1 + q_2(X_2 + … q_{l-1}(X_{l-1} + q_l(X'')))\\
  &+ q'_2(X'_2 + … q'_{l-1}(X'_{l-1} + q'_l(X'))))
\end{align*}
where $X' ⪯ X_l$ and $X'' ⪯ X_l$.
The operation copies the multiset marked by $q_2,…,q_l$ ``lossily'' to a multiset marked by $q_2', …, q_l'$.
Consider the configuration $q_1(q_2(q_3+q_3)+q_4(q_5))$ and the copy rule $(q_1,q_2)\cp(q_6,q_7)$.
The rule would select the part $q_2(q_3+q_3)$ and copy it, changing its label to $q_7$ and the control state to $q_6$ resulting in $q_6(q_2(q_3+q_3)+q_4( q_5)+q_7(q_3+q_3))$.

The operation $(q_1, …, q_l) \min (q_1', …, q'_l)$ can be seen as the inverse operation. Its semantics can be given by
\begin{align*}
    &q_1(X_1 + q_2(X_2 + … q_l(X_l)) + q'_2(X'_2 + … q'_l(X'_l)))\\
  → &q'_1(X_1 + q_2(X_2 + … q_{l-1}(X_{l-1})) + q'_2(X'_2 + … q'_{l-1}(X'_{l-1} + q'_l(X'))))
\end{align*}  
where $X' ⪯ X_l$ and $X' ⪯ X_l'$.
It deletes the multiset marked by $q_2, …, q_l$ and replaces the multiset marked by $q_2',…,q_l'$ with the minimum of both (or a smaller multiset).
Consider the configuration $q_1(q_2(q_3+q_4)+q_5(q_3+q_6))$ and the rule $(q_1,q_2)\min(q_7,q_5)$.
The rule would remove the part $q_2(q_3+q_4)$, replace $q_5(q_3+q_6)$ by the minimum and change the control state to $q_7$ resulting in $q_7(q_5(q_3))$.
Both operations can be implemented using standard NCS transition rules and do thus not extend the computational power of NCS.

\subparagraph{Implementing $\cp$.}
A copy rule $t = (q_1, …, q_l) \cp (q_1', …, q_l')$ can be implemented as follows: (The variables $r_i$ and index $q$ are universially quantified over all states)
\begin{align*}
  (q_1, …, q_l) &δ (\mathsf{cpi}_{t,q_l}, q_2, …, q_{l-1}, \mathsf{i})\\
  (\mathsf{cpi}_{t,q}, q_2, …, q_{l-1}) &δ (\mathsf{cpi}_{t,q}', q_2, …, q_{l-1}, \mathsf{o}_1)\\
  (\mathsf{cpi}_{t,q}', q_2', …, q_{l-1}') &δ (\mathsf{cp}_{t,q}, q_2', …, q_{l-1}', \mathsf{o}_2)\\
  (\mathsf{cp}_{t,q}, q_2, …, q_{l-1}, r_1, …, r_m, \mathsf{o}_1) &δ (\mathsf{cpd}_{t,q}, q_2, …, q_{l-1}, r_1, …, r_m, q, \mathsf{o}_1)\\
  (\mathsf{cpd}_{t,q}, q_2', …, q_{l-1}', r_1, …, r_m, \mathsf{o}_2) &δ (\mathsf{cpd}_{t,q}', q_2', …, q_{l-1}', r_1, …, r_m, q, \mathsf{o}_2)\\
  (\mathsf{cpd}_{t,q}', q_2, …, q_{l-1}, r_1, …, r_m, \mathsf{i}, r_{m+1}) &δ (\mathsf{cp}_{t,r_{m+1}}, q_2, …, q_{l-1}, r_1, …, r_m, q, \mathsf{i})\\
  (\mathsf{cp}_{t,q}, q_2, …, q_{l-1}, r_1, …, r_m, r_{m+1}, \mathsf{o}_1) &δ (\mathsf{cpu}_{t,q}, q_2, …, q_{l-1}, r_1, …, r_m, \mathsf{o}_1, q)\\
  (\mathsf{cpu}_{t,q}, q_2', …, q_{l-1}', r_1, …, r_m, r_{m+1}, \mathsf{o}_2) &δ (\mathsf{cpu}_{t,q}', q_2', …, q_{l-1}', r_1, …, r_m, \mathsf{o}_2, q)\\
  (\mathsf{cpu}_{t,q}', q_2, …, q_{l-1}, r_1, …, r_m, r_{m+1}, \mathsf{i}) &δ (\mathsf{cp}_{t,r_{m+1}}, q_2, …, q_{l-1}, r_1, …, r_m, \mathsf{i})\\
  (\mathsf{cp}_{t,q}, q_2, …, q_{l-1}, \mathsf{i}) &δ (\mathsf{cpf}_t, q_2, …, q_{l-1})\\
  (\mathsf{cpf}_t, q_2, …, q_{l-1}, \mathsf{o}_1) &δ (\mathsf{cpf}_t', q_2, …, q_l)\\
  (\mathsf{cpf}_t', q_2', …, q_{l-1}', \mathsf{o}_2) &δ (q_1', …, q_l')
\end{align*}
The construction works in a depth-first-search fashion using a symbol $\mathsf{i}$ to mark the set, that is currently copied (and subsequently deleted), and two symbols $\mathsf{o_1}$ and $\mathsf{o_2}$ to mark the two copies, that are currently created.
First (the control states named $\mathsf{cpi}$) the markings are placed.
Then either a new element of the multiset marked by $\mathsf{i}$ is selected, corresponding, new multisets are created under $\mathsf{o_1}$ and $\mathsf{o_2}$ and all markings are moved inwards ($\mathsf{cpd}$-states)
or copying of multiset marked by $\mathsf{i}$ has been completed, the multiset is deleted, and the markings are moved back outwards ($\mathsf{cpu}$-states).
When the markings are back on the outermost level, the copy process has been completed and the markings can be replaced ($\mathsf{cpf}$-states).

\subparagraph{Implementing $\min$.}
A minimum rule $t = (q_1, …, q_l) \min (q_1', …, q_l')$ can be implemented in a similar fashion:
\begin{align*}
  (q_1, …, q_l) &δ (\mathsf{mini}_{t,q_l}, q_2, …, q_{l-1}, \mathsf{i}_1)\\
  (\mathsf{mini}_{t,q}, q_2', …, q_l') &δ (\mathsf{mini}_{t,q}', q_2', …, q_{l-1}', \mathsf{i}_2)\\
  (\mathsf{mini}_{t,q}', q_2', …, q_{l-1}') &δ (\mathsf{min}_{t,q}, q_2', …, q_{l-1}', \mathsf{o})\\
  (\mathsf{min}_{t,q}, q_2', …, q_{l-1}', r_1, …, r_m, \mathsf{o}) &δ (\mathsf{mind}_{t,q}, q_2', …, q_{l-1}', r_1, …, r_m, q, \mathsf{o})\\
  (\mathsf{mind}_{t,q}, q_2, …, q_{l-1}, r_1, …, r_m, \mathsf{i}_1, r_{m+1}) &δ (\mathsf{mind}_{t,r_{m+1}}', q_2, …, q_{l-1}, r_1, …, r_m, q, \mathsf{i}_1)\\
  (\mathsf{mind}_{t,q}', q_2', …, q_{l-1}', r_1, …, r_m, \mathsf{i}_2, q) &δ (\mathsf{min}_{t,q}, q_2', …, q_{l-1}', r_1, …, r_m, q, \mathsf{i}_2)\\
  (\mathsf{min}_{t,q}, q_2', …, q_{l-1}', r_1, …, r_m, r_{m+1}, \mathsf{o}) &δ (\mathsf{minu}_{t,q}, q_2', …, q_{l-1}', r_1, …, r_m, \mathsf{o}, q)\\
  (\mathsf{minu}_{t,q}, q_2, …, q_{l-1}, r_1, …, r_m, r_{m+1}, \mathsf{i}_1) &δ (\mathsf{minu}_{t,r_{m+1}}', q_2, …, q_{l-1}, r_1, …, r_m, \mathsf{i}_1)\\
  (\mathsf{minu}_{t,q}', q_2', …, q_{l-1}', r_1, …, r_m, q, \mathsf{i}_2) &δ (\mathsf{min}_{t,q}, q_2', …, q_{l-1}', r_1, …, r_m, \mathsf{i}_2)\\
  (\mathsf{min}_{t,q}, q_2, …, q_{l-1}, \mathsf{i}_1) &δ (\mathsf{minf}_t, q_2, …, q_{l-1})\\
  (\mathsf{minf}_t, q_2', …, q_{l-1}', \mathsf{i}_2) &δ (\mathsf{minf}_t', q_2', …, q_{l-1}')\\
  (\mathsf{minf}_t', q_2', …, q_{l-1}', \mathsf{o}) &δ (q_1', …, q_l')
\end{align*}
It follows exactly the same idea, but deletes elements from two marked multisets ($\mathsf{i_1}$ and $\mathsf{i_2}$) and only creates elements in one marked multiset ($\mathsf{o}$).

\subparagraph{Hardy computations.}
Having these auxiliary operations at our disposal, we can now give the exact transition rules to implement Hardy computations.
The encoding of the ordinal parameter $α$ and the natural attribute $n$ of a Hardy function $H^α(n)$ is encoded into transitions as defined above.
We have to come up with transition rules that allow four kinds of runs
\begin{enumerate}
  \item $C_{α+1, n} →^* C_{α, n+1}$,
  \item $C_{α, n+1} →^* C_{α+1, n}$,
  \item $C_{α+λ, n} →^* C_{α+λ_n, n}$ and
  \item $C_{α+λ_n, n} →^* C_{α+λ, n}$
\end{enumerate}
in order to satisfy Lemma~\ref{lem:ncshardness1} without violating Lemma~\ref{lem:ncshardness2}.

Case (1) is straightforward, we only have to remove some element from the multiset encoding the ordinal and move it to the multiset encoding the argument:
\begin{align*}
  (\mathsf{main}, s, ω) &δ (\mathsf{R1}, s)\\
  (\mathsf{R1}, c) &δ (\mathsf{main}, c, 1)
\end{align*}
Case (2) works just the other way around:
\begin{align*}
  (\mathsf{main}, c, 1) &δ (\mathsf{R2}, c)\\
  (\mathsf{R2}, s) &δ (\mathsf{main}, s, ω)
\end{align*}

Case (3) requires to replace the smallest addend $ω^β$ of a limit ordinal $α + ω^β$ with the $n$th element of its fundamental sequence $(ω^β)_n$.
If $β$ is a limit ordinal, it has to be replaced by $β_n$, i.e. the same process has to be applied recursively.
Otherwise, the immediate predecessor of $β' + 1 = β$ has to be copied $n$ times.
The states in the following constructions are parametrised by the recursion depth~$l$.
\begin{align*}
  (\mathsf{main}, s, ω) &δ (\mathsf{R3}_0, s, a_1)\\
  (\mathsf{R3}_l, \overbrace{s', …, s'}^l) &δ (\mathsf{R3s}_l, \overbrace{s', …, s'}^l, s')\\
  (\mathsf{R3s}_l, s, \overbrace{s, …, s}^l, ω) &δ (\mathsf{R3s}_l', s, \overbrace{s, …, s}^l, a_2)\\
  (\mathsf{R3s}_l', s, \overbrace{s, …, s}^l, a_1) &\cp (\mathsf{R3s}_l'', s', \overbrace{s', …, s'}^l, ω)\\
  (\mathsf{R3s}_l'', s, \overbrace{s, …, s}^l, a_2) &\min (\mathsf{R3s}_l, s, \overbrace{s, …, s}^l, a_1)\\
  (\mathsf{R3s}_l, s, a_1, ω) &δ (\mathsf{R3}_{l+1}, s, s, a_1)\\
  (\mathsf{R3s}_l, s, \overbrace{s, …, s}^l, a_1, ω) &δ (\mathsf{R3c}_l, s, \overbrace{s, …, s}^l, a_1)\\
  (\mathsf{R3c}_l, s, \overbrace{s, …, s}^l, a_1) &\cp (\mathsf{R3c}_l', s', \overbrace{s', …, s'}^l, ω)\\
  (\mathsf{R3c}_l', c, 1) &δ (\mathsf{R3c}_l'', c)\\
  (\mathsf{R3c}_l'', c') &δ (\mathsf{R3c}_l, c, 1)\\
  (\mathsf{R3c}_l, c) &δ (\mathsf{R3q}_l)\\
  (\mathsf{R3q}_l, c') &δ (\mathsf{R3q}_l', c)\\
  (\mathsf{R3q}_l', s) &δ (\mathsf{R3q}_l'')\\
  (\mathsf{R3q}_l'', s', \overbrace{s', …, s'}^l) &δ (\mathsf{main}, s, \overbrace{ω, …, ω}^l)
\end{align*}
The construction starts selecting the smallest addend, by copying the multiset marked by $\mathsf{s}$ to the multiset marked by $\mathsf{s}'$ in descending order.
The descending order is ensured using the $\min$ operation introduced above.
$\mathsf{a}_1$ and $\mathsf{a}_2$ are used to mark the currently largest and second largest addend.
Once the copying process is stopped, $\mathsf{a}_1$ marks the supposedly smallest addend, the construction moves down one level, and repeats this process.
This part is implemented using the $\mathsf{R3s}$-states.
Once, a level is reached where the exponent is no longer a limit ordinal, one element is removed from the respective multiset (transition from $\mathsf{R3s}$ to $\mathsf{R3c}$).
Then, the copy operation is used to copy that exponent $n$ times.
This part is implemented by the $\mathsf{R3c}$-states.
Finally the multiset from the old ordinal is deleted and replaced by the newly computed ordinal ($\mathsf{R3q}$-states).
This construction might make several lossy errors in the sense that they result in a smaller ordinal to be computed.
E.g. it might not select the smallest addend at the cost of losing all smaller addends or it might stop at a level, where the exponent is still a limit ordinal. In this case instead of decreasing it by only one, a larger addend will be removed.

Case (4) can be handled similarly to (3).
The construction recursively guesses the smallest addend ($\mathsf{R4s}$-states) as before.
Then $n$ copies of an addend $ω^β$ have to be replaced by $ω^{β+1}$ ($\mathsf{R4m}$-states).
This is realised by deleting at most $n$ elements in descending order and maintaining their minimum using the minimum operation.
The construction counts the number of elements actually deleted and uses it as the new value for $n$, ensuring that a lossy error occurs in case less than $n$ elements are removed.
The exponent is then increased by one and the addend is moved to the newly created ordinal.
\begin{align*}
  (\mathsf{main}, s, ω) &δ (\mathsf{R4}_0, s, a_1)\\
  (\mathsf{R4}_l, \overbrace{s', …, s'}^l) &δ (\mathsf{R4s}_0, \overbrace{s', …, s'}^l, s')\\
  (\mathsf{R4s}_l, s, \overbrace{s, …, s}^l, ω) &δ (\mathsf{R4s}_l', s, \overbrace{s, …, s}^l, a_2)\\
  (\mathsf{R4s}_l', s, \overbrace{s, …, s}^l, a_1) &\cp (\mathsf{R4s}_l'', s', \overbrace{s', …, s'}^l, ω)\\
  (\mathsf{R4s}_l'', s, \overbrace{s, …, s}^l, a_2) &\min (\mathsf{R4s}_l, s, \overbrace{s, …, s}^l, a_1)\\
  (\mathsf{R4s}_l, s, a_1, ω) &δ (\mathsf{R4}_{l+1}, s, s, a_1)\\
  (\mathsf{R4s}_l, s, \overbrace{s, …, s}^l, ω) &δ (\mathsf{R4m}_l, s, \overbrace{s, …, s}^l, a_2)\\
  (\mathsf{R4m}_l, s, \overbrace{s, …, s}^l, a_2) &\min (\mathsf{R4m}_l', s, \overbrace{s, …, s}^l, a_1)\\
  (\mathsf{R4m}_l', c, 1) &δ (\mathsf{R4m}_l'', c)\\
  (\mathsf{R4m}_l'', c') &δ (\mathsf{R4m}_l, c, 1)\\
  (\mathsf{R4m}_l, s, \overbrace{s, …, s}^l, a_1) &δ (\mathsf{R4q}_l, s, \overbrace{s, …, s}^l, a_1, ω)\\
  (\mathsf{R4q}_l, s, \overbrace{s, …, s}^l, a_1) &\cp (\mathsf{R4q}_l', s', \overbrace{s', …, s'}^l, ω)\\
  (\mathsf{R4q}_l', c) &δ (\mathsf{R4q}_l'')\\
  (\mathsf{R4q}_l'', c') &δ (\mathsf{R4q}_l''', c)\\
  (\mathsf{R4q}_l''', s) &δ (\mathsf{R4q}_l'''')\\
  (\mathsf{R4q}_l'''', s', \overbrace{s', …, s'}^l) &δ (\mathsf{main}, s, \overbrace{ω, …, ω}^l)
\end{align*}

Finally, observe that for $α ≤ (Ω_k)_l$ the innermost (level-$k$) exponents of the CNF terms that can arise during the computation of $H^α(n)$ are bounded by $l$ because they are only decreased. 
Hence, using additional states $ω^0, …, ω^l∈Q$ on level $k-1$ to represent configurations $(ω,｛(ω,∅):i｝)$ by $(ω^i,∅)$ avoids one level of nesting in $𝓝$.

%% file: apx_freeze_linearisation.tex
\section{From LTL$^↓_{tqo}$ to Nested Counter Systems}
\label{apx:freeze2ncs}

In this section we provide the technical details of the reduction from the satisfiability problem for LTL$_A^↓$ formulae over tree-quasi-ordered attributes $A$ to the coverability problem in NCS.

\subsection{Reduction to Linear Orderings}

We recall and prove Proposition~\ref{prp:linearisation}.

\prplinearisation*

Let $Φ$ be an LTL$^↓_A$ formula. 
To translate $Φ$ into an equisatisfiable LTL$^↓_{[k]}$ formula for some $k∈ℕ$ we first turn $A$ into a \emph{tree-partial-order} $A'$ by collapsing maximal strongly connected components (SCC) and adjust $Φ$ to obtain an equisatisfiable formula $Φ'$ over  LTL$^↓_{A'}$.
Second, we show how to encode $A'$-attributed data words into $[k]$-attributed data words and translate $Φ'$ to operate on this encoding.

\subparagraph{Collapsing SCC.}

Let $C_{2,1},…,C_{2,n_2},C_{3,1},…,C_{3,n_3},…,C_{m,n_m}⊆A$ be all maximal strongly connected components in the graph of the tree-quasi-ordering $(A,⊑)$ of size larger than 1 such that $|C_{i,j}| = i$.
I.e., $C_{i,j}$ is the $j$-th distinct such component of size $i$. 
Notice that all $C_{i,j}$ are disjoint since they are maximal. 
Choose some arbitrary $x_{i,j}∈C_{i,j}$ from each component and remove all components from $A$ but for those elements $x_{i,j}$.
Thus we collapse all SCC in $A$ and obtain a tree-partial-ordering $A'$.
In the formula $Φ$ we syntactically replace every attribute $x∈C_{i,j}$ by the corresponding representative $x_{i,j}$ and obtain an LTL$^↓_{A'}$ formula.

Due to the semantics of the logic being defined in terms of downward closures the only significant change upon collapsing SCC is their size. 
While the downward-closures of two SCCs that have different sizes cannot be isomorphic replacing them with a single attribute can make valuations for them equal wrt. $≃$.
We therefore add the following constraint to $Φ$ disallowing a collapsed model to assign the same data value to representatives of SCCs that had different size.
\[
  ⋀_{x_{i,j}, x_{i',j'}∈A'｜i≠i'} \G(↓^{x_{i,j}}¬\F↑^{x_{i',j'}})
\]
Compared to the original models of $Φ$ this is not a restriction and thus every model of $Φ$ still induces a model of $Φ'$ and vice versa.

\subparagraph{Frame encoding.}

In the following we assume that $A$ is a tree(-partial)-ordering, i.e., it does not contain non-trivial SCCs.
Let $k$ be the depth of $A$, i.e., the length of the longest simple path starting at some root (minimal element).
We can pad $A$, by additional attributes s.t.\ \emph{every} maximal path in $A$ has length $k$.
The additional attributes added to $A$ this way are not smaller than the original ones and hence do not affect the semantics of formulae over $A$ except that the new attributes need to be assigned an arbitrary value.
Thus, regarding $A$ as a forest, we can assume that every leaf is at level $k$ (roots are at level 1).

Let $ℓ_1,…,ℓ_n∈A$ be the leafs in $A$ (enumerated in an \emph{in-order} fashion).
We use the ideas from \cite{DBLP:conf/fsttcs/KaraSZ10,DBLP:conf/concur/DeckerHLT14} to encode an $A$-attributed data word $w=w_1w_2…$ into a $[k]$-attributed data word $u=u_1u_2…$ where a single position in $w$ is represented by a \emph{frame} of $n$ positions in $u$.
Then, each position $w_i=(a_i,𝐝_i)$ in $w$ corresponds to the frame $u_{(i-1)n+1}…u_{(in)}$ in $u$.
In the $i$-th such frame, each position $u_{(i-1)n+j} = (a_i,𝐠_{i,j})$ carries the same letter $a_i$ as $w_i$.
The data valuation $𝐠_{(i,j)}∈Δ^{k}$ at the $j$-th position in the frame shall represents the $j$-th “branch” $𝐝_i|_{ℓ_j}$ of the valuation $𝐝_i$.
Thus, let for a leaf $ℓ_j$ in $A$ be $x_{j,1}⊑x_{j,2}⊑…⊑x_{j,k} =ℓ_j$ the attributes in $\cl(ℓ_j)$, representing the branch in $A$ from a root to $ℓ_j$.
Now for $r∈[k]$ we let $𝐠_{(i,j)}(r) = 𝐝_i(x_{j,r})$

\subparagraph{Translation.}

Based on this encoding we can translate any LTL$^↓_A$ formula $Φ$ to an LTL$^↓_{[k]}$ formula $\hat{Φ}$ that specifies precisely the encodings of models of $Φ$.
In particular, $\hat{Φ}$ is satisfiable iff $Φ$ is.

Given the (in-order) enumeration $ℓ_1,…,ℓ_n∈A$ of leafs in $A$ and an attribute $x∈A$ we let $\mathsf{sb}(x)=\min ｛r∈[n]｜x⊑ℓ_r｝$ denote the smallest index $r$ of a branch containing $x$ and $\mathsf{lb}(x)=\max ｛r∈[n]｜x⊑ℓ_r｝$ the largest such branch index.
Further, we denote by $\mathsf{lvl}(x)=|｛x'⊑x'｜x'∈A｝|$ its level in $A$.

We can assume $Φ$ to be in a normal form where every freeze quantifier $↓^x$ is followed immediately by either an $\X$, $\WX$ or $↑^y$ operator for attributes $x,y∈A$.
This is due to the following equivalences for arbitrary formulae $ψ,ξ$, letters $a∈Σ$ and attributes $x,y∈A$.
\[
\begin{array}{rcl@{\hspace{.5cm}}rcl}
  ↓^xa     &≡& a               &  ↓^x\Fψ  &≡& (↓^xψ)∨↓^x\X\Fψ \\
  ↓^x↓^yψ  &≡& ↓^yψ            & ↓^x\Gψ  &≡& (↓^xψ)∧↓^x\WX\Gψ \\
  ↓^x¬ψ    &≡& ¬↓^xψ           & ↓^x(ψ\Uξ) &≡& (↓^xξ) ∨ ((↓^xψ) ∧ ↓^x\X(ψ\Uξ))\\
  ↓^x(ψ∧ξ) &≡& (↓^xψ) ∧ (↓^xξ) & ↓^x(ψ\Rξ) &≡& (↓^xξ) ∧ ((↓^xψ) ∨ ↓^x\WX(ψ\Rξ))\\
  ↓^x(ψ∨ξ) &≡& (↓^xψ) ∨ (↓^xξ) 
\end{array}\]
We can further assume that for every formula $↓^x↑^y$ we have $\mathsf{sb}(x)≤\mathsf{sb}(y)$:
if $x⊏y$ or $x⊒y$ we can completely remove the formula, replacing it with a contradiction or a tautology, respectively.
Otherwise $x$ and $y$ are incomparable. 
Then, if $\mathsf{lvl}(x)<\mathsf{lvl}(y)$ the formula is again false and we can remove it.
For $\mathsf{lvl}(x)=\mathsf{lvl}(y)$ we have $↓^x↑^y  ≡ ↓^y↑^x $ and can swap them if necessary.
Finally, if $\mathsf{lvl}(x)>\mathsf{lvl}(y)$ there is a unique attribute $p⊏x$ with $\mathsf{lvl}(p)=\mathsf{lvl}(y)$ and by the definition of the semantics we have $↓^x↑^y  ≡ ↓^p↑^y$. 
We can thus replace $x$ by $p$ and swap the attributes if necessary.

Next we extend the alphabet to $Σ' = Σ×[n]$.
The attached number is supposed indicate the relative position in every frame.
This is enforced by a formula 
\[
  β_1:= Σ_i ∧ \G\left(⋀_{i∈[n]} Σ_i → \left((\WX Σ_{(i \mod n) +1}) ∧ ⋀_{j∈[n]∖｛i｝} ¬Σ_j\right)\right)
\]
where $Σ_i$ for $i∈[n]$ stands for the formula $⋁_{a∈Σ} (a,i)$.
Further, we impose that models actually have the correct structure and thereby encode an $A$-attributed data word.
The formula
\[
  β_2 := ⋀_{(a,i)∈Σ×[n-1]} \G((a,i) → \X (a,i+1))
\]
expresses that the letter from $Σ$ is constant throughout a frame and 
\[
β_3 := ⋀_{x∈A} \G \left(  Σ_1 → \X^{\mathsf{sb}(x)-1}↓^{\mathsf{lvl}(x)}\left(
    ↑^{\mathsf{lvl}(x)} \U (Σ_{\mathsf{lb}(x)} ∧ ↑^{\mathsf{lvl}(x)})
\right)  \right)
\]
ensures that the frame consistently encodes a valuation from $Δ^A$.
Finally, we define the translation $t(Φ)$ inductively for subformulae $ψ,ξ$ of $Φ$, $x∈A$ and $a∈Σ$ as follows.
\[\begin{array}{rcl@{\hspace{1cm}}rcl}
  t(↓^xψ) &:=& \X^{\mathsf{sb}(x)-1}↓^{\mathsf{lvl}(x)}t(ψ) &   t(a) &:=& a\\
  t(\Xψ) &:=& ⋀_{j=1}^n Σ_j → \X^{n-j+1}t(ψ) &   t(¬ψ) &:=& ¬t(ψ)\\
  t(ψ\Uξ) &:=& ((Σ_1 → t(ψ))\U(Σ_1∧t(ξ))) &   t(ψ ∧ ξ) &:=& t(ψ) ∧ t(ξ)\\
  t(↑^x) &:=& ⋀_{j=1}^n Σ_j → \X^{\mathsf{sb}(x)-j} ↑^{\mathsf{lvl}(x)}
\end{array}\]
We omit the remaining operators since they can be expressed in terms of the ones considered above.

To see that $\hat{Φ} := t(Φ) ∧ β_1 ∧ β_2 ∧ β_3$ exactly characterises the encodings of models of $Φ$ consider the underlying invariant that all subformulae of $Φ$ are always evaluated on the first position of a frame except those preceded by a freeze quantifier. 
Those that directly follow a freeze quantifier have the form $\Xψ$ or $↑^x$ and are  relocated to the first position of the successive frame or to the position encoding the branch of data values that needs to be checked, respectively.

%% file: apx_freeze2ncs.tex
\subsection{From LTL$^↓_{[k]}$ to NCS}

Recall Theorem~\ref{thm:freezetoncs}.

\thmfreezetoncs*

By Proposition~\ref{prp:linearisation} it suffices to show that given an LTL$^↓_{[k]}$ formula $Φ$ we can construct a $(k+1)$-NCS $𝓝$ and two configurations $C_{init},C_{final}∈𝓒_𝓝$ s.t.\ $Φ$ is satisfiable if and only if $C_{final}$ can be covered from $C_{init}$.

The idea is to construct from the LTL$^↓_{[k]}$ formula $Φ$ a $(k+1)$-NCS $𝓝$ that guesses an (abstraction of a) data word $w∈(Σ×Δ^k)^+$ position-wise starting with the last position and prepending new ones.
Simultaneously, $𝓝$ maintains a set of \emph{guarantees} for the so far constructed suffix of $w$.
These guarantees are subformulae $φ$ of $Φ$ together with an (abstraction of a) data valuation representing the register value under which $φ$ is satisfied by the current suffix of $w$.
Guarantees can be assembled to larger formulae in a way that maintains satisfaction by the current suffix of $w$.
Then $Φ$ is satisfiable if and only if there is a reachable configuration of $𝓝$ that contains $Φ$ as one of possibly many guarantees.

\subparagraph{Normal form.}

We fix for the rest of this section $k∈ℕ$ and an LTL$^↓_{[k]}$[\X,\WX,\U,\R] formula $Φ$ over the finite alphabet $Σ$ and the data domain $Δ$.
W.l.o.g.\ we restrict to the reduced set of temporal operators and expect $Φ$ to be in \emph{negation normal form}, i.e., negation appears only in front of letters $a∈Σ$ and check operators $↑^i$ for $i∈[k]$.
Further, we assume that every check operator $↑^i$ occurs within the scope of the freeze quantifier $↓^j$ of level $j≥i$ since otherwise the check necessarily fails and the formula can easily be simplified syntactically.
Let $\sub(Φ)$ denote the set of syntactical subformulae as well as the unfoldings of $\U$ and $\R$ formulae.

\subparagraph{State space.}

For the LTL$^↓_{[k]}[\X,\WX,\U,\R]$ formula $Φ$ we construct a $(k+1)$-NCS $𝓝_Φ = (Q,δ)$ as follows.
The state space is defined as 
$Q = Q_{\mathsf{ctrl}} ∪ Q_{\mathsf{cell}}$ where
\begin{align*}
   Q_{\mathsf{ctrl}} &= Q_\add ∪ Q_\nxt ∪ Q_\stp ∪ Q_{\st},\\
  Q_{\add} &= ｛\add｝× (Σ ∪ (Σ × \mathsf{sub}(Φ)),\\
  Q_{\nxt} &= ｛\nxt_1,\nxt_2, \cpy, \cpy_{bt}｝ ∪ (｛\cpy｝×2^{\mathsf{sub}(Φ)}),\\
  Q_{\stp} &= ｛\stp｝,\\
  Q_{\st} &= ｛\st, \st^\cmark, \aux, \aux^\cmark｝ \text{ and }\\
  Q_{\mathsf{cell}} &= ｛\cmark,\xmark｝ × 2^{\mathsf{sub}(Φ)}.
\end{align*}

The two outer-most levels (level 0 and 1) of configurations will only use states from $Q_\mathsf{ctrl}$ and control the management of the configurations of level 2 to $k$ below.
These configurations only use states from $Q_\mathsf{cell}$ and implement a \emph{storage} for a tree structure (more precisely, a forest) of depth $k$ represented by a multiset of configurations of level 2.
Every node in that forest, a \emph{cell}, stores a set of formulae and is checked (\cmark) or unchecked (\xmark).

Next we define the transition rules $δ⊆⋃_{i,j∈[k+1]}(Q^i×Q^j)$.

\subparagraph{Setup phase.}

The storage of the initial configuration 
\[
  C_{init} = \stp(\st(q_1(…q_{k-1}(q_k)…)))
\] 
with $q_1=…=q_k=(\cmark,∅)$ is empty except for a single checked branch of length $k$.
We allow the NCS to arbitrarily add new (unchecked) branches and then populate the branches with guarantees of the form $\WXφ∈\sub(Φ)$.
Thus, let
\begin{align*}
  (\stp, \st, q_1,…,q_i)&δ(\stp, \st, q_1,…,q_i,q_{i+1}',…,q_k')\\
  (\stp,\st,q_1,…,q_i,(m,F))&δ(\stp,\st,q_1,…,q_i,(m,F∪｛\WXφ｝)) \\
\end{align*}
for all $0≤i<k$, $q_1,…,q_i∈Q_{\mathsf{cell}}$, $q_{i+1}'=…=q_k'=(\xmark,∅)$,
$m∈｛\cmark,\xmark｝$, $F⊆\sub(Φ)$ and $\WXφ∈\sub(Φ)$.

\subparagraph{Construction phase.}

After the initial setup the NCS guesses a letter $a∈Σ$ by applying
\[
  (\stp)δ((\add, a)).
\]
New atomic formulae can be added by the rules
\begin{align*}
    ((\add,a),\st,q_1,…,q_i,(m,F)) & δ((\add,a),\st,q_1,…,q_i,(m,F∪｛a｝))\\
    ((\add,a),\st,q_1,…,q_i,(m,F)) & δ((\add,a),\st,q_1,…,q_i,(m,F∪｛¬b｝))\\
    ((\add,a),\st,q_1,…,q_{i+1}^\cmark,…,(m_j,F_j)) & δ((\add,a),\st,q_1,…,q_{i+1}^\cmark,…,(m_j,F_j∪｛↑^ℓ｝))\\
    ((\add,a),\st,q_1,…,q_i,(\xmark,F)) & δ((\add,a),\st,q_1,…,q_i, (\xmark,F∪｛¬↑^{ℓ'}｝))
\end{align*}
and existing formulae can be combined by rules
\begin{align*}
  ((\add,a),\st,q_1,…,q_i,(m,F∪｛φ｝))             & δ((\add,a),\st,q_1,…,q_i,(m,F∪｛φ,φ∨ψ｝))\\
  ((\add,a),\st,q_1,…,q_i,(m,F∪｛φ｝))             & δ((\add,a),\st,q_1,…,q_i,(m,F∪｛φ,ψ∨φ｝))\\
  ((\add,a),\st,q_1,…,q_i,(m,F∪｛φ,ψ｝))           & δ((\add,a),\st,q_1,…,q_i,(m,F∪｛φ,ψ,φ∧ψ｝))\\
  ((\add,a),\st,q_1,…,q_i,(m,F∪｛φ,ψ｝))           & δ((\add,a),\st,q_1,…,q_i,(m,F∪｛φ,ψ,ψ∧φ｝))\\
  ((\add,a),\st,q_1,…,q_i,(\cmark,F∪｛φ｝)) & δ((\add,a,↓^{i+1}φ),\st,q_1,…,q_i,(\cmark,F∪｛φ｝))\\
  ((\add,a,↓^{i+1}φ),\st,q_1,…,q_j,(m,F))  & δ ((\add,a),\st,q_1,…,q_j,(m,F∪｛↓^{i+1}φ｝))  
\end{align*}
\begin{align*}
    &((\add,a),\st,q_1,…,q_i,(m,F∪｛ψ∨(φ∧\X(φ\Uψ))｝))  \\
&\qquad  δ ((\add,a),\st,q_1,…,q_i,(m,F∪｛φ\Uψ｝))\\
&     ((\add,a),\st,q_1,…,q_i,(m,F∪｛ψ∧(φ∨\WX(φ\Rψ))｝))\\
&\qquad  δ  ((\add,a),\st,q_1,…,q_i,(m,F∪｛φ\Rψ｝))
\end{align*}
for, respectively,
$F,F_j⊆\sub(Φ)$,
$m,m_j∈｛\cmark,\xmark｝$,
$0≤i,j<k$, $ℓ∈[i+1]$, $i<ℓ'≤k$,
$q_1,…,q_k∈Q_{\mathsf{cell}}$,
$q_{i+1}^\cmark=(\cmark,F)$,
$b∈Σ∖｛a｝$ and 
$φ,ψ,φ∨ψ,ψ∨φ, φ∧ψ, ψ∧φ,↓^{i+1}φ,↑^ℓ, a, ¬b, φ\Uψ, φ\Rψ∈\sub(Φ)$.

\subparagraph{Advancing phase.}

To ensure consistency, prepending of $\X$ and $\WX$ operators can only be done for all stored formulae at once.
This corresponds to guessing a new position in a data word, prepending it to the current one and computing a set of guarantees for that preceeding position from the guarantees of the current position.

The NCS can enter the advancing phase by the rules
\[
  ((\add,a))δ(\nxt_1, \aux^\cmark)
\]
for $a∈Σ$.
This also creates an auxiliary storage.
Next, the original storage is copied cell by cell to the auxiliary storage.
Upon copying a cell the formulae stored within are preceeded by next-time operators.
To this end, for $F⊆\sub(Φ)$ we denote by 
$F_{\X} =｛\Xφ,\WXφ∈\sub(Φ)｜φ∈F｝$.

The markings are now utilised as pointers to the cell currently being copied.
The rules 
\[
  (\nxt_1 ,\st,q_1^\cmark,…,q_k^\cmark) δ (\cpy, \st^\cmark ,q_1^\xmark,…,q_k^\xmark)
\]
for $q_1^\cmark=(\cmark,F_1),…,q_k^\cmark=(\cmark, F_k)$, $q_1^\xmark=(\xmark,F_1),…,q_k^\xmark=(\xmark, F_k)$,  where $F_1,…,F_k⊆\sub(Φ)$, 
set these pointers to the root of the storage 

To allow the NCS to copy the cells over in a depth-first, \emph{lossy} manner let
\begin{align*}
  (\cpy,\st^\cmark,(\xmark,F)) & δ((\cpy,F),\st,(\cmark,F))\\ 
  ((\cpy,F),\aux^\cmark) & δ(\cpy,\aux,(\cmark,F_{\X}))\\  
  (\cpy,\st,q_1,…,q_i,(\cmark,F'), (\xmark,F)) & δ((\cpy, F), \st,q_1,…,q_i,(\xmark,F'),(\cmark,F))\\
  ((\cpy,F),\aux,q_1,…,q_i,(\cmark,F')) & δ (\cpy,\aux,q_1,…,q_i,(\xmark,F'),(\cmark,F_{\X}))
\end{align*}
for, respectively, $F,F'⊆\sub(Φ)$, $0≤i<k$ and $q_1…,q_i∈Q_{\mathsf{cell}}$.
To allow for backtracking we let
\begin{align*}    
  (\cpy,\st,q_1,…,q_i,(\xmark,F),(\cmark,F')) &δ (\cpy_{bt},\st,q_1,…,q_i,(\cmark,F)\\
  (\cpy_{bt},\aux,q_1,…,q_i,(\xmark,F),(\cmark,F')) &δ (\cpy,\aux,q_1,…,q_i,(\cmark,F),(\xmark,F'))\\
  (\cpy,\st,(\cmark,F)) &δ (\cpy_{bt},\st^\cmark)\\
  (\cpy_{bt},\aux,(\cmark,F)) &δ (\cpy,\aux^\cmark,(\xmark,F))
\end{align*}
for $0≤i<k$, $F,F'⊆\sub(Φ)$ and $q_1,…,q_i∈Q_{\mathsf{cell}}$.

Finally the original storage can be replaced by the auxiliary one by
\begin{align*}
  (\cpy,\st^\cmark)&δ(\cpy_{bt})\\
  (\cpy_{bt},\aux^\cmark)&δ(\nxt_2,\st)
\end{align*}
The storage is now (partially) copied to the auxiliary storage.
To enter the construction phase and thereby complete the transition from the old position in the imaginary data word to the preceeding position a new checked branch and a new letter from $Σ$ is guessed by
\[
  (\nxt_2,\st,(\xmark,F_1),…,(\xmark,F_i)) δ((\add,a),\st,(\cmark,F_1)…,(\cmark,F_k))
\]
for any $a∈Σ$, $0≤i≤k$, $F_1,…,F_i⊆\sub(Φ)$ and $F_{i+1}=…=F_k=∅$.

\subsection{Correctness}

The NCS $𝓝=(Q,δ)$ that construct above maintains a forest of depth $k$ where every node is labelled by a set of subformulae of $Φ$.
Configurations reachable from the initial configuration 
\[
  C_{init} = \stp(\st(q_1(…q_{k-1}(q_k)…)))
\] 
with $q_1=…=q_k=(\cmark,∅)$ always have the form $q_{ctrl}(q_\st(X)+X')$ or $\cpy_{bt}(\aux^\cmark(X))$
where $X$ is a multiset of configurations of level 2 that represents a forest $T_C$ of depth $k$.%
Let $V_C$ be the set of nodes and $F:V_C → 2^{\sub(Φ)}$ their labelling by sets of formulae.
For a node $v∈V_C$ at level $i$ (roots have level 1) in $T_C$ we denote by $ρ(v)=v_1…v_i$ the unique path from a root to $v_i=v$.

The structure of the forest represents the context in which the individual formulae are assumed to be evaluated.
To formalise this let $con: V_C → Δ$ be a labelling of $T_C$ by data values called \emph{concretisation}.
Such a labelling induces a set $G_{con}(T_C)⊆\sub(Φ)×Δ^k$ with $(φ,𝐝)∈G_{con}(T_C)$ iff 
\begin{itemize}
  \item $φ∈F(v)$ for some node $v∈V_C$ with $ρ(v) = v_1…v_j$ and 
  \item $𝐝∈Δ^j$ with $𝐝(i)=con(v_i)$ for $i∈[j]$.
\end{itemize}

Now, let $w=(a,𝐝)u∈(Σ×Δ^k)^+$ be a data word.
We say that $w$ and a configuration $C=(q_0,M)$ are \emph{compatible} if and only if there is a concretisation $con: V → Δ$ such that 
\begin{itemize}
  \item for all $(φ,𝐝')∈G_{con}(T_C)$ we have $(w,1,𝐝')⊧φ$ (guarantees are satisfied), 
  \item $q_0∉｛(\add,b),(\add,b,φ)｜b∈Σ∖｛a｝, φ∈\sub(Φ)｝$ (letter is compatible) and
  \item if $q_0∉Q_{\nxt}$ and $v_1…v_k$ is the unique path in $T_C$ corresponding to the checked cells in $C$ 
        then $(con(v_1),…,con(v_k)) = (𝐝(1), …,𝐝(k))$ (valuation is compatible).
\end{itemize}

\begin{lemma}[Invariant]
  Let $C →C'$ be two configurations reachable from $C_{init}$ and $w∈(Σ×Δ^k)^+$
  a $[k]$-attributed data word such that $T_C$ and $w$ are compatible. 
  Then there is a $[k]$-attributed data word $w'∈(Σ×Δ^k)^+$ such that $T_{C'}$ and $w'$ are compatible.
\end{lemma}

\begin{proof}
The initial configuration $C_{init}$ does not contain any guarantees and hence every data word is compatible with $T_{C_{init}}$.
The only formulae added during the setup phase are of the form $\WXφ$.
Thus, every configuration reachable during this phase is compatible at least with every data word of length 1.

Consider a configuration $C$ in the construction phase being compatible with a data word $w∈(Σ×Δ^k)^+$ due to a concretisation $con$.
It is easy to see that the atomic formulae that can be added are satisfied on $w$ under the same concretisation $con$.
Also, the Boolean combinations of satisfied formulae that can be added remain satisfied and the folding of temporal operators respects the corresponding equivalences.

A rule adding a formula $↓^iφ$ can obviously only be executed if $φ$ is present in the marked cell at level $i$.
Since in particular the valuation $𝐝$ of the first position of $w$ is compatible with the marking there is $(φ,𝐝|_{i})∈G_{con}(T_C)$ and $(w,1,𝐝|_{i})⊧φ$.
Hence $(w,1,𝐝')⊧↓^iφ$ for any valuation $𝐝'$ and $↓^iφ$ can be put into any cell in $T_C$ without breaking compatibility with $w$ under $con$.

Consider a configuration $C$ in the advancing phase. 
The transition rules staying in the phase do not add any new formula to any cell in the storage of the configuration and hence any word compatible with $C$ remains compatible.

Finally, assume that $w$ is compatible with a configuration $C$ due to a concretisation $con: V_C → Δ$ and that a transition rule of the form
\[
  (\nxt_2,\st,(\xmark,F_1),…,(\xmark,F_i)) δ((\add,a),\st,(\cmark,F_1)…,(\cmark,F_k)),
\]
for $a∈Σ$, $0≤i≤k$, $F_1,…,F_i⊆\sub(Φ)$ and $F_{i+1}=…=F_k=∅$,
is applied to obtain a configuration $C'$.

Let $(a',𝐝')$ be the first position of $w$ and $d_{i+1}, …,d_k∈Δ∖\mathsf{img}(con)$ data values that are not assigned to any node in $T_C$ by $con$.
We define a new valuation $𝐝∈Δ^k$ such that $𝐝|_{i} = 𝐝'|_{i}$ and 
$𝐝(j) = d_j$ for $i<j≤k$.
Then the word $(a,𝐝)w$ is compatible with $C'$ witnessed by the concretisation $con':V_{C'} → Δ$ with 
$con'(v) =con(v)$ for nodes $v∈V_C$ that were already present in $T_C$ and $con'(v_j') = d_j$ for the new nodes $v_j$ ($i<j≤k$) created by the rule.
\end{proof}

As a consequence of the previous lemma we conclude that if a configuration $C$ containing $Φ$ as guarantee is reachable from the initial configuration $C_{init}$ then it is satisfiable.
We allow the NCS to enter a specific target state $q_{final}$ once the formula $Φ$ is encountered somewhere in the current tree. 
Thus, a path covering $C_{final} = q_{final}$ proves $Φ$ satisfiable.
Conversely, if $Φ$ has some model $w$ than the NCS $𝓝$ as constructed above can guess according to the letters and valuations along the word and assemble $Φ$ from its subformulae.

%% file: apx_ncs2freeze.tex
\section{From NCS to LTL$^↓_{tqo}$}
\label{apx:ncs2freeze}

We provide the detailed construction to prove Theorem~\ref{thm:ncstofreeze}.

\thmncstofreeze*

 Let $𝓝=(Q,δ)$ be a $k$-NCS.
We are interested in describing witnesses for coverability. 
It suffices to construct a formula $Φ_𝓝$ that characterises precisely those words that encode a \emph{lossy} run from some configuration $C_{start}$ to some configuration $C_{end}$.
We call a sequence $C_0C_1…C_n$ of configurations $C_j∈𝓒_𝓝$ a lossy run from $C_0$ to $C_n$ if there is a sequence of intermediate configurations $C_0'…C_{n-1}'$ such that $C_i⪰C_i'→C_{i+1}$ for $0≤i<n$, i.e.,
\[
C_0⪰C_0'→C_1⪰ C_1'→…→C_{n-1}⪰C_{n-1}'→C_n.
\]
Then $C_{end}$ is coverable from $C_{start}$ if and only if there is a lossy run from $C_{start}$ to some $C_n⪰C_{end}$. 
For $𝓝$ we construct a formula
\[
  \mathrm{Φ_𝓝 = Φ_{conf} ∧ Φ_{flow} ∧ Φ_{rn} ∧ Φ_{inc} ∧ Φ_{dec} ∧ Φ_{start} ∧ Φ_{end}}.
\]
where
\begin{itemize}
  \item $\mathrm{Φ_{conf}}$ describes the shape of a word to encode a sequence of configurations,
  \item $\mathrm{Φ_{flow}}$ enforces that in addition to the plain sequence of encoded configurations there are annotations that indicate which transition rule is applied and which part of configuration is affected by the rule,
  \item $\mathrm{Φ_{rn}}$, $\mathrm{Φ_{inc}}$, $\mathrm{Φ_{dec}}$, encode the correct effect of transition of the respective type (see below),
  \item $\mathrm{Φ_{start}}$ encodes that the encoded run starts with $C_{start}$ and
  \item $\mathrm{Φ_{end}}$ encodes that the encoded run ends with some configuration $C⪰C_{end}$.
\end{itemize}

We omit to construct $\mathrm{Φ_{start}}$ and $\mathrm{Φ_{end}}$ since it is straightforward given the considerations below.
For easier reading we use an alphabet of the form $Σ=2^{AP}$ where $AP$ is a set of \emph{atomic propositions}.
Formally, every proposition $p∈AP$ used in a formula below could be replaced by 
\[
  ⋁_{a∈Σ｜p∈a} a
\]
 to adhere to the syntax defined in Section~\ref{sec:intro}.

\subsection{Configurations}

A configuration $C=(q,M)∈𝓒_𝓝$ of some $k$-NCS $𝓝=(Q,δ)$ can be interpreted as a tree $T$ of depth at most $k+1$ where the root carries $q$ as label. 
The children of the root are the subtrees $T_{(1,1)},…,T_{(1,n)}$ represented by the configurations of level 1 contained in the multiset $M$.

Similar to the approach of Proposition~\ref{prp:linearisation} we encode such a tree as $[k]$-attributed data word.
We use an alphabet $Σ$ where every letter $a∈Σ$ encodes a $(k+1)$-tuple of states from $Q$, i.e., a possible branch in the tree.
Then a sequence of such letters represents a set of branches that form a tree.
The data valuations represent the structure of the tree, i.e., which of the branches share a common prefix.
Two branches represented by positions $(a,𝐝), (a',𝐝')∈Σ×Δ^k$ are considered to be identical up to level $0≤i≤k$ if and only if $(𝐝(1), …,𝐝(i))=(𝐝'(1), …,𝐝'(i))$. 
Notice that the tuples of states represented by $a$ and $a'$ must also coincide on their $i$-th prefix if (but not only if) $𝐝$ and $𝐝'$ do.

For technical reasons we require that positions are arranged such that in between two positions representing branches with a common prefix of length $i$ there is no position representing a branch that has a different prefix of length $i$.
Further, this representation is interlaced: it refers only to odd positions.
The even position in between are used to represent an exact copy of the structure but uses different data values.
An example is shown in Figure~\ref{fig:ncs-conf-encoding}.

We specify the shape of data words that encode (sequences of) configurations by the following formulae.
For convenience we assume w.l.o.g.\ that the NCS uses a distinct set of states $Q_i⊆Q$ for each level $0≤i≤k$ that includes an additional state $-_i∈Q_i$ not occurring in any transition rule from $δ$.

\begin{itemize}
  \item Positions are marked by even/odd.
  \[
      (\mathsf{odd} → \WX(¬\mathsf{odd} ∧ \WX \mathsf{odd} )) ∧ (\mathsf{even}↔ ¬\mathsf{odd})
  \]
  \item Data values are arranged in blocks. Once a block ends, the respective value (valuation prefix) will never occur again (on an odd position).
    \begin{gather}\label{eqn:unique-values}
      ⋀_{i∈[k]} ↓^k ((\X\X¬↑^i) ∧ \mathsf{odd}) → ¬\X\F(\mathsf{odd}\ ∧ ↑^i)
    \end{gather}
  \item Positions in the same block on level $i$ carry the same states up to level $i$.
    \[
      ⋀_{i∈[k]} ⋀_{q∈Q_i} (q ∧ ↓^k\X\X↑^i) → \X\X q
    \]
  \item Even positions mimic precisely the odd positions but use different data values.
    \begin{equation}\label{eqn:mimic-odds}
       \mathsf{odd}  → ((⋀_{q∈Q} q ↔ \X q) ∧ (↓^1\X¬↑^1) ∧ ⋀_{i∈[k]}↓^k\X\X↑^i ↔  \X(↓^k\X\X↑^i))
    \end{equation}
  \item State propositions are obligatory and mutually exclusive on every level.
    \begin{gather*}
      \left(⋀_{0≤i≤k} ⋁_{q∈Q_i} q\right) \\
      ∧ ⋀_{0≤i≤k}\ ⋀_{q,q'∈Q_i｜q≠q'} q → ¬q'
    \end{gather*}
  \item Branches shorter than $k+1$ are padded by states $-_i∈Q$.
    \begin{equation} \label{eqn:dummy-states}
    \begin{split}
      \Big(&⋀_{i∈[k-1]} -_i → -_{i+1}\Big)\\
      ∧ &⋀_{i∈[k-1]} (↓^k\X\X↑^i) → ¬-_{i+1} ∧ \X\X¬-_{i+1}
    \end{split}
    \end{equation}
  \item The proposition \$ is used to mark the first position of every configuration and can thus only occur on odd positions and the beginning of a new block.
  \[
    (\$ → \mathsf{odd}) ∧ ((↓^1\X\X↑^1) → \X\X¬\$)
  \]
  \item Freshness propositions mark only positions carrying a valuations of which the prefix of length $i$ has not occurred before.
  \[
  ⋀_{i∈[k]} ↓^i ¬\X\F(↑^i ∧\ \textsf{fresh}_i)
  \]
\end{itemize}
Let $φ$ be the conjunction of these constraints and $\mathrm{Φ_{conf}}:= \$ ∧ \Gφ$.

\subsection{Control Flow}

To be able to formulate the effect of transition rules without using past-time operators we encode the runs \emph{reversed}. 
Given that a data word encodes a sequence $C_0C_1…C_n$ of configurations as above we model the (reversed) control flow of the NCS $𝓝=(Q,δ)$ by requiring that every configuration but for the last be annotated by some transition rule $t_j∈δ$ for $0≤j<n$.
\begin{equation}\label{eqn:rule-annotation}
  \G((\$∧\X\F\$) ↔ ⋁_{t∈δ} t)
\end{equation}

Now, the following constraints impose that this labelling by transitions actually represents the reversal of a lossy run.
That is, for every configuration $C_j$ in the sequence (for $0≤j<n$) with annotated transition rule $t_j$ there is a configuration $C'$ (not necessarily in the sequence) such that $C'\stackrel{t}{→}C_j$ and $C'⪯C_{j+1}$.

\subparagraph{Marking rule matches.}
Consider a position $C_j$ in the encoded sequence that is annotated by a transition $t_j=((q_0,…,q_i),(q_0',…,q_j'))∈δ$.
In order for $C_j$, $t_j$ and $C_{j+1}$ to encode a correct (lossy) transition first of all there must be a branch in $C_j$ that matches $(q_0',…,q_j')$.
We require that one such branch is marked by propositions $\cmark_1 … \cmark_j$:
\begin{multline*}
  ⋀_{t=((q_0,…,q_i),(q_0',…,q_j'))∈δ}\G\Big( t →\ ((\X¬\$)\U (\mathsf{even} ∧ ⋀_{ℓ∈[j]}\cmark_ℓ ∧ q_ℓ'))\\
                                               ∧ (q_0' ∧ ¬\cmark_{j+1} ∧…∧ ¬\cmark_k)\U\$\Big)
\end{multline*}

This selects a branch of length $j$ in every configuration. 
An operation that affects this branch may also affect other branches sharing a prefix.
Thus, they are supposed to be marked accordingly.
Since a node at some level $i∈[k]$ in the configuration tree is encoded by a block of equal data valuations on level $i$, blocks are marked entirely or not at all.

\[
⋀_{i∈[k]}\G( ↓^k\X\X↑^i → (\cmark_i ↔ \X\X\cmark_i))
\]
Moreover, at most one block is marked in every configuration frame.
\[
⋀_{i∈[k]}\G\big( (\cmark_i ∧ ↓^k\X\X¬↑^i)  → \X((¬\cmark_i)\U\$)\big)
\]
Now the markers indicate which positions in the word are affected by their respective transition rule.
Notice, that the even positions are supposed to carry the marking.
Let $\mathrm{Φ_{flow}}$ be the conjunction of the three formulae above and that in Equation~\ref{eqn:rule-annotation}.

\subsection{Transition Effects}

It remains to assert the correct effect of each transition rule to the marked branches.
We distinguish three rule types.
Let 
\[
  δ_{rn} = ｛((q_0,…,q_i),(q_0',…,q_i'))∈δ｜0≤i≤k｝
\] 
be the set of \emph{renaming} transition rules, 
\[
  δ_{dec} = ｛((q_0,…,q_i,q_{i+1},…,q_{j}),(q_0',…,q_i'))∈δ｜0≤i<j≤k｝
\]
be the set of \emph{decrementing} transition rules and
\[
  δ_{inc} = ｛((q_0,…,q_i),(q_0',…,q_i',q_{i+1}',…,q_j'))∈δ｜0≤i<j≤k｝
\]
be the set of \emph{incrementing} transition rules.
Then $δ = δ_{rn} ∪ δ_{dec} ∪ δ_{inc}$.

\subparagraph{Renaming rules.}
Let
\begin{equation}\label{eqn:renaming-rule}
  \mathrm{Φ_{rn}} = 
    ⋀_{t=((q_0,…,q_i),(q_0',…,q_i'))∈δ_{rn}} \G\left( t → \textsf{copy}(q_0,…,q_i) \right).
\end{equation}
The formula specifies that whenever a configuration $C_n$ is supposed to be obtained from a configuration $C_{n+1}$ using a renaming transition rule $t∈δ_{rn}$ then every branch in $C_n$ is also present in $C_{n+1}$.
Moreover the states $q_0',…,q_i'$ in the marked branch should be replaced by $q_0,…,q_i$.

The idea to realise this is to use the interlaced encoding of configurations to link (identify) branches from a configuration $C_n$ with the configuration $C_{n+1}$ in the sequence represented by a potential model.
We consider a branch in $C_n$ linked (on level $i∈[k]$) to a branch in $C_{n+1}$ if the corresponding even position in $C_n$ and the corresponding odd position in $C_{n+1}$ carry the same valuation (up to level $i$).
Since valuations uniquely identify a particular block (Equation~\ref{eqn:unique-values}) the following formula enforces for a position that there is a corresponding position in a unique block at level $i$ of the next configuration where additionally $φ$ is satisfied.
\[
  \textsf{link}_i(φ) =↓^k((¬\$) \U \Big(\$ ∧ \big((\X¬\$)\U (↑^i ∧\ \textsf{odd} ∧\ φ)\big)\Big)).
\]
For $i=k$ a block represents an individual branch in a configuration and the formula $\textsf{link}_k(φ)$ enforces that there is a unique corresponding branch in the consecutive configuration.

The even positions in turn mimic the odd ones using different data values (Equation~\ref{eqn:mimic-odds}).
Thereby we can create a chain of branches that are linked and thus identified.
We use this to enforce that for renaming transition rules, each branch present in a configuration $C_n$ will again occur in $C_{n+1}$ ensuring that the sequence is “gainy” wrt. the branches---and hence lossy when being reversed.

Using this we define the $\textsf{copy}$ formula in Equation \ref{eqn:renaming-rule} as
\begin{multline*}
\textsf{copy}(q_0,…,q_i) =  \left( \textsf{even} → \left(
  \begin{split}
    \textsf{link}_k(q_0)\ ∧ \big(⋀_{ℓ∈[i]}(\cmark_ℓ → \textsf{link}_k(q_ℓ)) \big) \\
     ∧ ⋀_{ℓ∈[k], q∈Q_ℓ}(¬\cmark_ℓ ∧ q) → \textsf{link}_k(q)
  \end{split}
      \right)\right)\U \$
\end{multline*}

\subparagraph{Decrementing rules.} We address this case by the formula
\begin{equation}
  \mathrm{Φ_{dec}} = 
    ⋀_{t=((q_0,…,q_j),(q_0',…,q_i'))∈δ_{dec}} \G\left( t → 
    \left(\begin{split}
      \textsf{copy}(q_0,…,q_i) \\ ∧\ \textsf{new}_i(q_0,…,q_j) 
    \end{split}
    \right)\right)
\end{equation}
where $i<j$ and
\[%
\textsf{new}_i(q_0,…,q_j) = (¬\X\$) \U (\cmark_i ∧ \textsf{link}_i(\textsf{fresh}_{i+1} ∧ q_0 ∧ … ∧ q_j))
\]
ensures that a configuration $C_{n+1}$ actually contains a branch that can be removed by a decrementing rule $t∈δ_{dec}$ rule to obtain $C_n$.

\subparagraph{Incrementing rules.} 
For the remaining case let 
\begin{equation}
  \mathrm{Φ_{inc}} = 
    ⋀_{t=((q_0,…,q_i),(q_0',…,q_j'))∈δ_{inc}} \G\left( t → 
    \textsf{copyBut}_j(q_0,…,q_i) ∧ \textsf{zero}_i
    \right)
\end{equation}
where $i<j$,
\[
  \textsf{zero}_j = (\X¬\$) \U (\cmark_j ∧ ⋀_{j<ℓ≤k} -_ℓ)
\]
asserts that the new branch that is created by an incrementing rule $t∈δ_{inc}$ does not contain more states than explicitly specified in $t$.
Recall that Equation~\ref{eqn:dummy-states} ensures that the propositions $-_ℓ$ for $ℓ∈[k]$ can only appear if there are no actual states below level $ℓ-1$ in the tree structure of the corresponding configuration.

Finally, the formula
\begin{multline*}
\textsf{copyBut}_j(q_0,…,q_i) = \\
  \left( (\textsf{even} ∧ ¬\cmark_j) → \left(
  \begin{split}
    \textsf{link}_k(q_0)\ ∧
       \big(⋀_{ℓ∈[i]}(\cmark_ℓ → \textsf{link}_k(q_ℓ)) \big) \\
     ∧ ⋀_{ℓ∈[k], q∈Q_ℓ}(¬\cmark_ℓ ∧ q) → \textsf{link}_k(q)
  \end{split}
       \right)\right)\U \$.
\end{multline*}
is similar to the $\textsf{copy}$ formula above but omits to copy the particular branch that was created by the incrementing rule.